\documentclass[letterpaper,twocolumn,11pt,unpublished]{quantumarticle}
\pdfoutput=1
\usepackage{verbatim}
\usepackage[utf8]{inputenc}
\usepackage[english]{babel}
\usepackage[table]{xcolor}
\usepackage[T1]{fontenc}
\usepackage{amsmath}
\usepackage{amssymb}
\usepackage{hyperref}
\usepackage{tikz}
\usepackage{lipsum}
\usepackage{graphicx}
\usepackage{epsfig}
\usepackage{epstopdf}
\usepackage{braket}
\usepackage{hyperref}
\usepackage{amsthm}
\usepackage{enumitem}
\usepackage{dsfont}
\usepackage{bbold}
\usepackage{bm}
\usepackage{eurosym}
\usepackage{diagbox}
\usepackage[numbers,sort&compress]{natbib}
\usepackage[ruled,vlined]{algorithm2e}
\usepackage{comment}

\DeclareMathOperator{\N}{\mathbb{N}}
\DeclareMathOperator{\Z}{\mathbb{Z}}

\newcommand{\be}{\begin{equation}}
\newcommand{\ee}{\end{equation}}
\newcommand{\ba}{\begin{aligned}}
\newcommand{\ea}{\end{aligned}}

\newcommand{\R}{\mathbb{R}}
\newcommand{\bc}{\begin{center}}
\newcommand{\ec}{\end{center}}
\newcommand{\beq}{\begin{equation}}
\newcommand{\eeq}{\end{equation}}
\newcommand{\beqq}{\begin{equation*}}
\newcommand{\eeqq}{\end{equation*}}
\newcommand{\beqa}{\begin{align}}
\newcommand{\eeqa}{\end{align}}
\newcommand{\barr}{\begin{array}}
\newcommand{\earr}{\end{array}}
\newcommand{\bi}{\begin{itemize}}
\newcommand{\ei}{\end{itemize}}

\newcommand{\C}{\mathbb{C}}

\newcommand*{\coloneqq}{\mathrel{\vcenter{\baselineskip0.5ex \lineskiplimit0pt \hbox{\scriptsize.}\hbox{\scriptsize.}}} =}

\newtheorem{lem}{Lemma}
\newtheorem{theo}{Theorem}

\newtheorem{defi}{Definition}
\newtheorem*{lem*}{Lemma}

\DeclareMathOperator{\poly}{poly}

\DeclareMathOperator{\Tr}{Tr}

\begin{document}

\title{Exponentially-improved effective descriptions of physical bosonic systems}

\author{Varun Upreti}
\email{varun.upreti@inria.fr}
\affiliation{DIENS, \'Ecole normale sup\'erieure, PSL University, CNRS, INRIA, 45 rue d’Ulm, Paris, 75005, France}
\author{Nicolás Quesada}
\orcid{0000-0002-0175-1688}
\affiliation{Département de génie physique, École polytechnique de Montréal, Montréal, QC, H3T 1J4, Canada}
\author{Ulysse Chabaud}
\email{ulysse.chabaud@inria.fr}
\orcid{0000-0003-0135-9819}
\affiliation{DIENS, \'Ecole normale sup\'erieure, PSL University, CNRS, INRIA, 45 rue d’Ulm, Paris, 75005, France}

\maketitle

\begin{abstract}
The effective description of a bosonic quantum system identifies the minimum finite dimension required 
to capture its essential dynamics. This effective dimension plays an important role in the complexity of classical and quantum algorithms for learning and simulating bosonic systems. While generic 
bosonic states require a dimension scaling as $1/\epsilon^2$ for a precision of approximation $\epsilon$, 
here we identify a natural energy condition which allows us to improve this scaling exponentially
to $\log(1/\epsilon)$. We then prove that most bosonic quantum states satisfy this condition, and in particular those produced by combining Gaussian dynamics with generic energy-preserving dynamics, which include the output states of universal bosonic quantum circuits. We apply this finding to enhance learning algorithms for bosonic quantum states and we further obtain new classical simulation algorithms for a large class of bosonic systems. Finally, using efficient decompositions of Kerr gates as sums of Gaussian gates, we significantly refine these classical simulation algorithms for universal bosonic quantum circuits. Our results demonstrate that physical bosonic systems are significantly 
more well-behaved than previously assumed, allowing for efficient descriptions even at high precision.
\end{abstract}

\section{Introduction}

The physics of many natural systems is governed by bosonic degrees of freedom, leading to hallmark phenomena such as black-body radiation \cite{bose1924plancks}, superfluidity in liquid Helium \cite{London1938} and the Hong--Ou--Mandel effect \cite{Hong1987}. In addition, bosonic platforms have recently emerged as a promising avenue for building quantum computers \cite{Gottesman2001,Knill2001,Menicucci2006,madsen2022quantum}, given their ability to generate complex quantum states \cite{yokoyama2013ultra,larsen2025integrated} and remarkable error-correction capabilities \cite{michael2016new,terhal2020towards}. In this context, classical or quantum algorithms for simulating bosonic systems become important: not only to better understand the physics of these systems \cite{Yalouz2021encodingstrongly,Tong2022provablyaccurate,hanada2025quantumbosons}, but also to pinpoint the conditions for ``quantum computational speedup'' using them \cite{Mari2012,chabaud2023resources,Dias2024,calcluth2025,upreti2025interplay}. 

However, the classical and quantum simulation of bosonic systems poses a challenge unique to this type of physical system: bosons live in an infinite-dimensional Hilbert space, and this makes their simulation using classical or finite-dimensional qubit-based quantum computers challenging. Therefore, for simulation purposes, one may look for an ``effective dimension'' of the system, for instance by truncating the Fock space (or boson number space) of the bosonic quantum states such that the finite-dimensional truncated state still capture the essential physics of the system \cite{arzani2025}. If the effective dimension of a bosonic quantum system is small enough throughout a given bosonic dynamics, one can then efficiently simulate the dynamics on a classical computer by tracking the description of the truncated state throughout. Efficient descriptions of the initial quantum state evolving under a bosonic Hamiltonian also becomes important when trying to simulate these Hamiltonian dynamics using a qubit-based quantum computer, as underlined in \cite{Tong2022provablyaccurate,hanada2025quantumbosons}.

Beyond classical and quantum simulation, the effective dimension of bosonic quantum states emerges as a key bottleneck for their efficient learning. In particular, energy-based boson number cutoffs for generic bosonic states require an effective dimension scaling as $1/\epsilon^2$ with respect to the desired precision $\epsilon$ between the original and the truncated state, thereby leading to inefficient learning of bosonic states \cite{mele2024learning}. 

Here we show that for most physical systems, we can do exponentially better, with boson number cutoffs scaling only as $\log(1/\epsilon)$. We do so by introducing exponential-energy-based cutoffs for bosonic states which lead to this improvement, and our main technical result is to prove that bosonic states produced by universal bosonic quantum circuits have bounded exponential-energy.

Our results have immediate applications for the learning theory of bosonic quantum states. Specifically, we demonstrate that learning physical bosonic states is significantly more efficient than suggested by previous bounds \cite{mele2024learning}, as these states admit much more efficient effective descriptions than previously known.

\begin{figure*}[t]
    \centering
    \includegraphics[width=\linewidth]{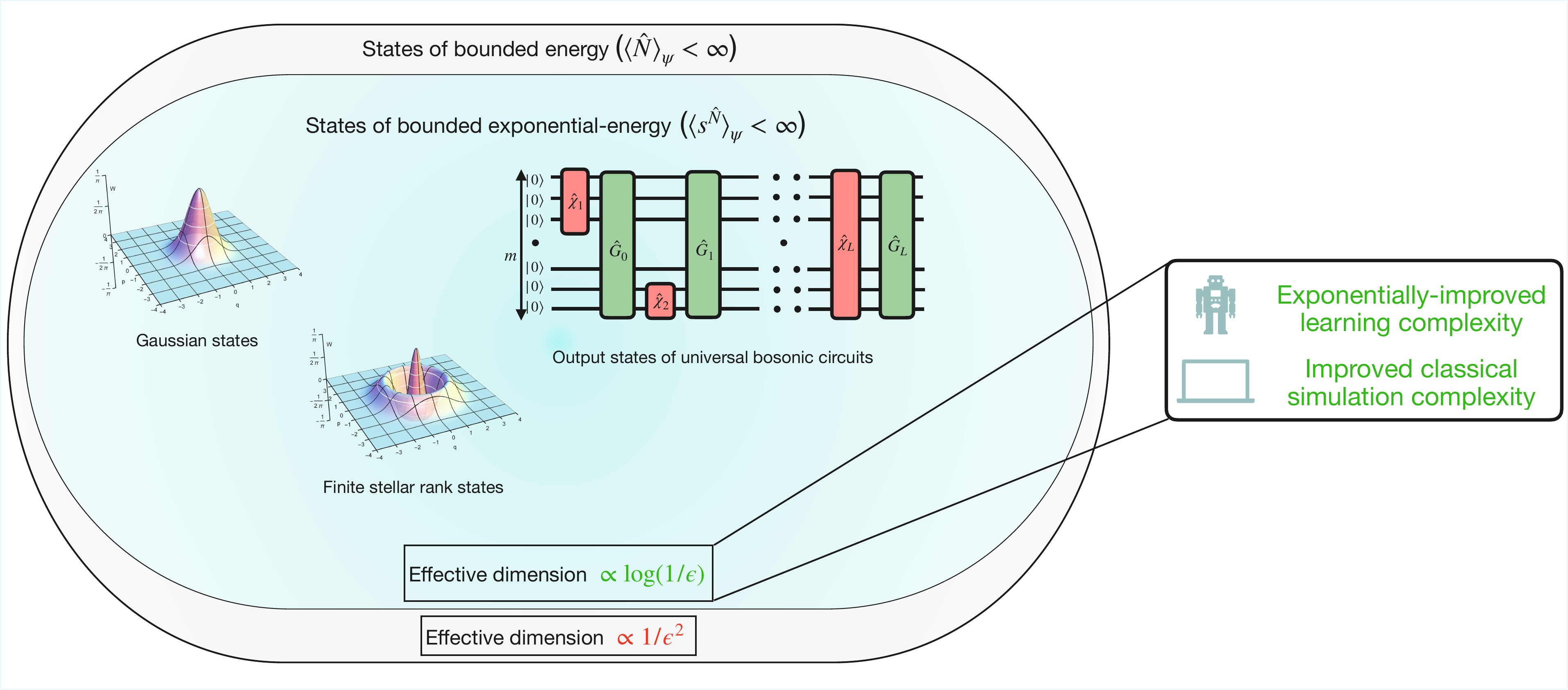}
    \caption{\textbf{Exponentially reduced effective dimension of physical bosonic quantum states.} While energy-based cutoffs for bosonic states lead to an effective dimension scaling as $1/\epsilon^2$ \cite{mele2024learning} in terms of the desired precision $\epsilon$, we consider the set $\mathcal S$ of states with bounded exponential-energy (Definition \ref{defi:bounded_exp_energy_states}) as a physically relevant class of bosonic quantum states, and show that the effective dimension of any $\ket{\psi} \in \mathcal S$ scales as $\log(1/\epsilon)$ (Lemma \ref{lem:efficient_description}). This set is known to include Gaussian states \cite{mele2025communication} and single-mode states of finite stellar rank \cite{cerf2025}, and our main technical contribution is to prove that output states of universal bosonic circuits also belong to $\mathcal S$ (Theorem \ref{theo:exp_energy_bound_simplified}). Combining these results yields an exponential improvement in effective state descriptions of physical bosonic quantum states, with applications to substantially more efficient learning of bosonic states (Theorem \ref{theo:precision_efficient_learning}) and classical simulation of bosonic circuits (Theorems \ref{theo:Fock_state_simulation} and \ref{theo:kerr_gate_simul}).}
    \label{fig:summary}
\end{figure*}

Furthermore, we provide new classical algorithms for simulating bosonic quantum circuits composed of Gaussian and energy-preserving non-Gaussian gates, either by tracking the succinct classical description of the evolved state in the Fock space throughout, made possible by the improved effective description, or, focusing specifically on circuits composed of Gaussian and Kerr gates, by combining our main result with decompositions of self-Kerr and cross-Kerr gates as sums of Gaussian gates. Since bosonic circuits composed of Gaussian and Kerr gates are universal in the standard model of bosonic (continuous-variable) quantum computing \cite{Lloyd1999}, our results highlight the potential of these improved effective descriptions for classical simulation of bosonic computations. 


\section{Results}\label{sec:results}

Our results are structured as follows (see Figure \ref{fig:summary}).
In Section \ref{sec:boundedexpE}, we define the set of states with bounded exponential-energy (Definition \ref{defi:bounded_exp_energy_states}) and we show that these states have exponentially reduced effective dimension in terms of the desired precision (Lemma \ref{lem:efficient_description}). Then, in Section \ref{sec:univ}, we prove that the states produced by universal bosonic quantum circuits have bounded exponential-energy (Theorem \ref{theo:exp_energy_bound_simplified}). Finally, we demonstrate the consequences of our results for learning bosonic quantum states (Theorem~\ref{theo:precision_efficient_learning}) and classically simulating universal bosonic computations (Theorems \ref{theo:Fock_state_simulation} and \ref{theo:kerr_gate_simul}) in Section \ref{sec:applications}.

We start by introducing the necessary preliminary material in the next Section \ref{sec:prelim}. Additional techniques used in the work, including an approximate decomposition of Kerr gates as sums of Gaussian gates based on Diophantine approximation of potential independent interest, are reported in \hyperref[sec:methods]{Methods}.

\subsection{Preliminaries}
\label{sec:prelim}
We refer the reader to \cite{NielsenChuang} for background on quantum information theory and to \cite{Braunstein2005,ferraro2005,Weedbrook2012} for continuous-variable (CV) quantum information material. Hereafter, the sets $\N, \R$, $\Z$ and $\C$ are the set of natural, real, integers and complex numbers respectively, with a * in the exponent when 0 is removed from the set and with a `+' in the subscript to refer to the non-negative subset of that particular set, when applicable.

In CV quantum information, a mode refers to a specific degree of freedom of a CV quantum system, and is the equivalent of a qubit in the CV regime. In this paper, $m \in \N^*$ denotes the number of modes in the system.

The single-mode creation and annihilation operator are denoted by $\hat a^\dagger$ and $\hat a$ respectively, and they satisfy $[\hat a, \hat a^\dagger] = \hat a \hat a^\dagger - \hat a^\dagger \hat a =\mathbb{I}$. The multimode number operator is given by 
\begin{equation}
    \hat N = \sum_{i=1}^m \hat a_i^\dagger \hat a_i = \sum_{i=1}^m \hat N_i.
\end{equation}
We denote by $\ket{\bm n}$ the multimode Fock state, with $\bm n = (n_1,\dots,n_m) \in \mathbb{N}^m$. These are eigenstates of the number operator: $\hat N\ket{\bm n} = |\bm n|\ket{\bm{n}}$, where we denote
\begin{equation}
    |\bm n| := \sum_{i=1}^m n_i.
\end{equation}
The set of Fock states $\{\ket{\bm n}\}_{\bm n \in \N^m}$ forms a basis, therefore any state $\ket{\psi}$ can be written as 
\begin{equation}
    \ket{\psi} = \sum_{\bm n} \psi_{\bm n} \ket{\bm n}.
\end{equation}
Furthermore, the energy of state $\ket{\psi}$ can be written as the expectation value of its multimode number operator as 
\begin{equation}
    E(\psi) = \bra{\psi} \hat N \ket{\psi},
\end{equation}
where we have dropped the state-independent zero point energy factor for simplicity. For any $s>1$, we also refer to the operator $s^{\hat N}$ as an exponential-energy operator.

The following result provides a generic bound on the effective dimension of energy-bounded bosonic quantum states, as a direct consequence of the gentle measurement lemma \cite{Winter1999,AnnaMele2024}:
\begin{lem}[Energy-based truncation of bosonic quantum states] \label{lem:energy_cutoff} Given a quantum state $\ket{\psi}$ such that $\bra{\psi}\hat N\ket{\psi} \leq E$, then by denoting $\ket{\psi_k}$ its normalized projection onto the space with total boson number $\leq k$, it follows that
\begin{equation}
   \frac12 \left\|\ket{\psi}\!\bra{\psi} -\ket{\psi_k}\!\bra{\psi_k}\right\|_1 \leq \sqrt{\frac{E}{k}}.
\end{equation}
Therefore, for boson number cutoff $k = E/\epsilon^2$, we have
\begin{equation}
     \frac12\left\|\ket{\psi}\!\bra{\psi} -\ket{\psi_k}\!\bra{\psi_k}\right\|_1 \leq\epsilon.
\end{equation}
\end{lem}

This result implies that, for any quantum state with bounded energy, picking an effective dimension scaling as $1/\epsilon^2$ and truncating the state ensures a precision of $\epsilon$ in trace distance. Note that a similar bound holds if the truncation is performed without renormalizing \cite{Winter1999}.

Finally, we introduce standard CV quantum operations. Products of unitary operations generated by Hamiltonians that are quadratic in the creation and annihilation operators of the modes are called Gaussian unitary operations, and states produced by applying a Gaussian unitary operation to the vacuum state are Gaussian states. These include passive Gaussian unitary gates such as beamsplitters and phase-shifters that preserve the boson number, the multimode displacement gate $\hat D(\bm{\alpha})$, with $\bm \alpha = (\alpha_1,\alpha_2,\dots,\alpha_m) \in \mathbb{C}^m$, and the multimode squeezing operator $\hat S(\bm\xi)$ with $\bm\xi = (\xi_1,\xi_2,\dots,\xi_m) \in \mathbb{C}^m$ \cite{BUCCO}. Furthermore, by the Bloch Messiah decomposition \cite{ferraro2005}, any Gaussian unitary $\hat G$ can be decomposed as
\begin{equation}\label{eq:Bloch_Messiah}
    \hat G = \hat V \hat D(\bm{\alpha}) \hat S (\bm r) \hat U,
\end{equation}
where $\bm \alpha \in \C^m$, $\bm r \in \R_+^m$, $\hat V$ and $\hat U$ are passive Gaussian unitary gates.

In their seminal paper on continuous-variable quantum computing \cite{Lloyd1999}, Lloyd and Braunstein suggested Gaussian and cubic phase gates, as well as Gaussian and Kerr gates, as universal gate sets for bosonic computations. While most follow-up work in bosonic computations has focused on circuits composed of Gaussian and cubic phase gates \cite{Sefi2011,Marshall2015,Miyata2016,Alexander2018,zheng2023gaussian,Budinger2024}, recent results have shown that energy growth in these circuits makes them unphysical \cite{chabaud2025energy} and that Kerr gates perform better at implementing logical operations on some encoded qubits~\cite{hastrup2021unsuitability,boudreault2026using}. Hereafter, we thus focus on universal circuits composed of Gaussian and energy-preserving non-Gaussian gates to analyze physically realizable bosonic computations. Prominent examples of energy-preserving (commuting with the number operator) non-Gaussian gates are the self-Kerr gate $\hat \kappa (x) = \exp(i\pi x \hat N^2)$ and the cross-Kerr gate $\hat{c\kappa}(x) = \exp(i2\pi x\hat N_1\otimes\hat N_2)$. 

\subsection{States with bounded exponential-energy}
\label{sec:boundedexpE}

There are several ways of defining ``physical'' bosonic states in the literature, with common definitions including states with bounded energy, states with bounded higher order energy moments, or states in the Schwartz space \cite{hall2013quantum}. In this work, we focus on the following set of bosonic states:
\begin{defi}[Bounded exponential-energy states]\label{defi:bounded_exp_energy_states}
    We define the set of $m$-mode bounded exponential-energy states $\mathcal S$ as
    \begin{equation}
        \mathcal S = \{\ket{\psi}: \bra{\psi}s^{\hat N}\ket{\psi} < \infty  \text{ for some }s > 1\}.
    \end{equation}
\end{defi}
\noindent Note that $\mathcal S$ is a dense subset of the set of bosonic quantum states (for the trace norm), since it contains all finite superpositions of Fock states. Other examples of states belonging to this set include multimode Gaussian states \cite{mele2025communication} and single-mode quantum states of finite stellar rank \cite{cerf2025}. As we prove in the next section, quantum states generated by universal bosonic circuits composed of Gaussian and energy-preserving non-Gaussian gates also belong to $\mathcal S$. 

We now discuss some properties of states in $\mathcal S$. Firstly, the wavefunctions of single-mode states $\ket{\psi} \in \mathcal S$ are particularly well-behaved, as they extend to holomorphic functions over the complex plane \cite{cerf2025}. In addition, since
\begin{equation}
    \bra{\psi}s^{\hat N}\ket{\psi} = \sum_{\bm{n}} s^{|\bm n|} |\psi_{\bm n}|^2,
\end{equation}
where $\bm n = (n_1,n_2,\dots,n_m) \in \mathbb{N}^m$, $|\bm n| = \sum_{k=1}^m n_k$ and $\psi_{\bm n}$ are the Fock state amplitudes of $\ket{\psi}$, the exponential-energy condition $\bra{\psi}s^{\hat N}\ket{\psi} < \infty$ implies an exponential decay in the Fock state amplitudes of a state $\ket{\psi} \in \mathcal S$. A natural way to decompose the set $\mathcal S$ is
\begin{equation}
    \mathcal S = \cup_{s > 1, E \in \R_{+}} \mathcal S_{s,E},
\end{equation}
with the set of states $\mathcal S_{s,E}$ defined as
\begin{equation}\label{eq:set_bounded}
\mathcal S_{s,E} = \{\ket{\psi} :\bra{\psi} s^{\hat N}\ket{\psi} \leq s^{E}\}.
\end{equation}
In particular, $\ket{\psi} \in \mathcal S$ implies that $\ket{\psi} \in \mathcal S_{s,E}$ for some $s > 1, E \in \R_+$. Jensen's inequality further implies that $s^{\bra{\psi}\hat N\ket{\psi}}\le \bra{\psi}s^{\hat N}\ket{\psi}$, so states in $\mathcal S_{s,E}$ also have energy bounded by $E$.
 
The exponential decay of Fock state amplitudes of bosonic states leads to an exponentially smaller effective dimension in terms of the desired precision:
\begin{lem}[Exponential-energy-based truncation of bosonic quantum states]\label{lem:efficient_description}
    Given a quantum state $\ket{\psi} \in \mathcal S_{s,E}$, the normalized Fock state truncation of $\ket{\psi}$ up to total boson number $k$, $\ket{\psi_k}$, is such that
    \begin{equation}
        \frac12 \|\ket{\psi}\!\bra{\psi} - \ket{\psi_k}\!\bra{\psi_k}\|_1 \leq \sqrt{\frac{s^E}{s^k}}.
    \end{equation}
    Therefore, choosing
    \begin{equation}\label{eq:boson_number_cutoff}
        k =E + 2\log_s(1/\epsilon),
    \end{equation} 
    we obtain
    \begin{equation}
        \frac12 \|\ket{\psi}\bra{\psi} - \ket{\psi_k}\bra{\psi_k}\|_1 \leq \epsilon.
    \end{equation}
\end{lem}
\noindent The proof of Lemma \ref{lem:efficient_description} is given in Section \ref{appsec:efficient_description} of the Supplementary Material, and consists of upper bounding the (in)fidelity $1 - |\braket{\psi_k|\psi}|^2$ in terms of $s,k$ and $E$.

Lemma \ref{lem:efficient_description} leads to an exponential improvement of the effective dimension of states with bounded exponential-energy with respect to the desired precision $\epsilon$, as compared to standard energy-based truncation (see Lemma \ref{lem:energy_cutoff}), from $1/\epsilon^2$ to $\log(1/\epsilon)$.

\subsection{Output states of universal bosonic circuits have bounded exponential-energy}
\label{sec:univ}

In this section, we prove the main technical result of our work, namely that the states produced by universal bosonic circuits have bounded exponential-energy. We first formally define this set of states:
\begin{defi}[Output states of universal bosonic circuits]  \label{defi:G_nG}
    We define the set of output states of universal $m$-mode bosonic circuits as follows:
    \begin{equation}
        \mathcal C_L  = \{\ket{\psi}\!:\!\ket{\psi} = \hat G_L \hat{\chi}_L \hat G_{L-1}\dots\hat{G}_1 \hat{\chi}_1\ket{0}^{\otimes m}\},
    \end{equation}
    where $L\ge1$ denotes the depth of the circuit, $\hat \chi_1,\dots,\hat \chi_L$ are energy-preserving non-Gaussian unitary gates, and $\hat G_1,\dots,\hat G_L$ are Gaussian unitary gates.
\end{defi}
\noindent Note that this class of states encompasses universal circuits composed of Gaussian gates and Kerr gates \cite{Lloyd1999}.
As discussed in Section \ref{sec:prelim}, the focus on universal bosonic circuits consisting of Gaussian and energy-preserving non-Gaussian gates is motivated by recent results showing that the energy growth in bosonic computation consisting of Gaussian and generic non-Gaussian gates have unrealistic energy growth \cite{chabaud2025energy}. 

We now state our main technical result, with a detailed version given in \hyperref[sec:methods]{Methods}:
\begin{theo}[Output states of universal bosonic circuits have bounded exponential-energy]\label{theo:exp_energy_bound_simplified}
For all $L\ge1$, $\mathcal C_L\subset\mathcal S$.
In particular, given an $m$-mode quantum state $\ket{\psi} \in \mathcal C_L$, it is possible to pick $t_L > 1, E_L \in \R_+$ such that $\ket{\psi} \in\mathcal S_{t_L,E_L} \subset \mathcal S$, or equivalently
\begin{equation}\label{eq:exp_energy_bound_informal}
    \bra{\psi}t_L^{\hat N}\ket{\psi} \leq t_L^{E_L}.
\end{equation}
\end{theo}
\noindent The proof of Theorem \ref{theo:exp_energy_bound_simplified} is given in the Supplementary Material \ref{appsec:exp_energy_bound energy simplified} and follows from the combination of results tracking the increase of the expectation value of the exponential-energy operator $s^{\hat N}$ of an arbitrary quantum state under $m$-mode squeezing and displacement operators.

From Lemma \ref{lem:efficient_description}, for a state $\ket{\psi} \in \mathcal C_L$ with $\bra{\psi} t_{L}^{\hat N}\ket{\psi} < t_L^{E_L}$, a boson number cutoff $k = E_L + 2\log_{t_L}(1/\epsilon)$ suffices to get an $\epsilon$-precise description of $\ket{\psi}$. The term $E_L$ in the exponent of $t_L$ in Eq.\ \ref{eq:exp_energy_bound_informal} can be understood as a proxy for the energy of $\ket{\psi}$ (see Section \ref{appsec:cutoff_comparisons} of the Supplementary Material for a detailed discussion). Therefore, comparing this with the boson number cutoff one gets for $\ket{\psi}$ through standard energy-based truncation $E_L/\epsilon^2$ (see Lemma \ref{lem:energy_cutoff}), the combination of Theorem \ref{theo:exp_energy_bound_simplified} and Lemma \ref{lem:efficient_description} allows for an exponential improvement in effective dimension from $1/\epsilon^2$ to $\log(1/\epsilon)$ in terms of the desired precision $\epsilon$. In Section \ref{appsec:cutoff_comparisons} of the Supplementary Material, we also provide a detailed comparison of the scaling of the boson number cutoff with respect to $L$, $m$, and $\epsilon$, highlighting the exponential improvement from our result over the bound derived from standard energy-based truncation.

\subsection{Applications}
\label{sec:applications}

In this section, we show the implications of the results from the previous section for the learning theory of bosonic states, and how they lead to new classical simulation algorithms for universal bosonic quantum circuits.

\subsubsection{Learning bosonic quantum states with bounded exponential-energy}
In \cite[Theorem S25]{mele2024learning}, it is shown that for pure bosonic quantum states $\ket{\psi}$ with $\bra{\psi}\hat N \ket{\psi} \leq mE$, a sample complexity (up to logarithmic factors) of
\begin{equation}
    \Theta \left(\frac{1}{\epsilon^2}\left(\frac{E}{\epsilon^2}\right)^m\right)
\end{equation}
is both necessary and sufficient to construct the classical description of an estimate $\ket{\tilde \psi}$ such that
\begin{equation}
    \frac12 \|\ket{\psi}\!\bra{\psi} - \ket{\tilde \psi}\!\bra{\tilde \psi}\|_1 \leq \epsilon.
\end{equation}
This makes learning of bosonic quantum states highly inefficient as compared to qudits, where the sample complexity scales as $\epsilon^{-2}$ as compared to $\epsilon^{-2m-2}$ for bosonic quantum states \cite{anshu2023surveycomplexity}. Conceptually, this happens because the energy bound does not guarantee the fast decay of Fock state amplitudes for higher boson number and as such, the heavy-tailed nature of these states precludes their efficient learning. In addition, in \cite[Theorem S25]{mele2024learning}, an improvement in the efficiency of learning bosonic states is obtained given an upper bound on the expectation value of $\hat N^k$ for the bosonic state, where $k > 1$. Therefore, for bounded exponential-energy states (see Definition \ref{defi:bounded_exp_energy_states}) which feature an exponential decay of Fock state amplitudes, one would intuitively expect that learning should be more efficient. We formalize this intuition with the following result:
\begin{theo}[Learning bosonic states with bounded exponential-energy]\label{theo:precision_efficient_learning}
    Given a state $\ket{\psi}$ belonging to the set $\mathcal S_{s,mE}$ given by Eq.~\ref{eq:set_bounded}, a number of samples 
    \begin{equation}
        \mathcal{O}\left(\frac{1}{\epsilon^2}\left(E + \frac{2\log_s(1/\epsilon)}{m}\right)^m\right),
    \end{equation}
    where $\mathcal O$ represents scaling up to logarithmic factors, is sufficient to construct the classical description of an estimate $\ket{\tilde \psi}$ such that 
    \begin{equation}
        \frac12 \|\ket{\psi}\bra{\psi} - \ket{\tilde \psi}\bra{\tilde \psi}\|_1 \leq \epsilon.
    \end{equation}
\end{theo}
\noindent The proof of Theorem \ref{theo:precision_efficient_learning} is given in Section \ref{appsec:precision_efficient_learning} of the Supplementary Material, and follows by finding the effective dimension of states with an exponential-energy bound, and using \cite[Theorem 1]{anshu2023surveycomplexity}.

From Theorem \ref{theo:precision_efficient_learning}, we note that while learning remains inefficient due to the exponential scaling of sample complexity with the number of modes $m$, the sample complexity is exponentially reduced from $\mathcal O(1/\epsilon^{2m+2})$ to $\mathcal{O}(1/\epsilon^2\times \log(1/\epsilon)^m)$, bringing the sample complexity of learning bosonic quantum states in $\mathcal S$ closer to that of qudits. Moreover, Theorem \ref{theo:exp_energy_bound_simplified} shows that the sets $\mathcal C_L$ (Definition \ref{defi:G_nG}) of output states of universal bosonic circuits are contained in $\mathcal S$. As such, our findings indicate that the complexity of learning physical bosonic quantum states in terms of desired precision is substantially more efficient than previously identified.


\subsubsection{Classical simulation of universal bosonic quantum circuits}
The exponentially reduced effective dimension for states in $\mathcal C_L$, produced by $L$ layers of Gaussian and energy-preserving non-Gaussian gates, allows for classical simulation by simply tracking the Fock state amplitudes of the evolved state throughout. One such simulation result is given in the following theorem, with a detailed version given in \hyperref[sec:methods]{Methods}:
\begin{theo}[Simulating bosonic computations with bounded exponential-energy]\label{theo:Fock_state_simulation}
Let $n$ be a size parameter. Consider an $m$-mode vacuum state $\ket{0}^{\otimes m}$ evolving through $L$ layers of Gaussian and energy-preserving non-Gaussian gates, followed by a Gaussian measurement. With $m = \mathcal{O}(\log(n))$ and $L = \mathcal{O}(\poly(n))$, and assuming that exponential-energy remains bounded as $\mathcal{O}(\poly(n))$ throughout the computation, output probabilities and marginal probabilities can be estimated classically up to $\mathcal{O}(1/\poly(n))$ additive precision in time $\mathcal{O}(\poly(n))$.
\end{theo}
\noindent The proof of Theorem \ref{theo:Fock_state_simulation} is given in Section \ref{appsec:Fock_space_simulation} of the Supplementary Material and consists in combining Theorem \ref{theo:exp_energy_bound_simplified} with Lemma \ref{lem:efficient_description} to find the effective dimension of the evolved state throughout the computation, and computing a Fock basis description of the output state. \cite[Theorem 2]{chabaud2021corestates} then allows for computing the probabilities generated by Gaussian state projections on the output state up to inverse polynomial precision in polynomial time. We also detail why a similar method (tracking Fock state amplitudes of the evolved state) would not be efficient with a standard energy-based cutoff (see Lemma \ref{lem:energy_cutoff}), because it would incur a super-polynomial overhead for the effective dimension of the state.

While Theorem \ref{theo:Fock_state_simulation} holds for generic Gaussian and energy-preserving non-Gaussian gates, it restricts the computations to be on the subsystems of a larger system for the classical simulation to be efficient. Efficient classical simulation algorithms can be devised without this restriction by exploiting the structure of specific gates. 

Hereafter, we focus on universal bosonic circuits where the energy-preserving non-Gaussian gates are either self-Kerr gates
\begin{equation}
    \hat \kappa(x) = \exp(i\pi x\hat N_1^2)
\end{equation}
or cross-Kerr gates
\begin{equation}
    \hat {c\kappa}(x) = \exp(i2\pi x \hat N_1 \otimes \hat N_2).
\end{equation}
We assume self-Kerr (resp.\ cross-Kerr) gates to be acting on the first mode (resp.\ the first two modes) without loss of generality up to to SWAP gates, which can be subsequently absorbed in Gaussian gate layers. 
Exploiting the specific structure of Kerr gates, we obtain the following result, with its formal version given in \hyperref[sec:methods]{Methods}:
\begin{theo}[Simulating universal bosonic computations with Kerr gates]\label{theo:kerr_gate_simul}
Consider an $m$-mode vacuum state $\ket{0}^{\otimes m}$ evolving through $L$ layers of Gaussian and Kerr gates, followed by a Gaussian measurement. With $L = \mathcal{O}(\poly(m))$, output probabilities and marginal probabilities can be computed: (i) exactly in polynomial time if all the Kerr gates have rational parameters, (ii) up to inverse-polynomial additive precision in quasi-polynomial time, for generic Kerr gates.
\end{theo}

\noindent The proof of Theorem \ref{theo:kerr_gate_simul} is given in Section \ref{appsec:kerr_gate_simul} of the Supplementary Material and is based on (approximate) decompositions of self-Kerr and cross-Kerr gates as sums of Gaussian gates. Self-Kerr gates with rational parameters, i.e.\ $\kappa(\frac pq)$ for $p \in \Z,q \in \Z_+^*$, can be decomposed exactly as a sum of $q$ (Gaussian) phase-shifters \cite{Tara1993,Jun-wei1996}. In Lemma \ref{lem:kerr_decomp} in \hyperref[sec:methods]{Methods}, we extend these decompositions to cross-Kerr gates with rational parameters, and we further give approximate decompositions of generic Kerr gates (with possibly irrational parameters) as sums of phase shifters, using Diophantine approximation \cite{hurwitz1891ueber} to approximate irrational parameters by rational ones. The precision of these approximate decompositions is governed by the effective dimension of the state that the self-Kerr or cross-Kerr gate is being applied to, where we can leverage Theorem \ref{theo:exp_energy_bound_simplified} and Lemma \ref{lem:efficient_description} (see Section \ref{methodsec:phase_shifter_simul} in \hyperref[sec:methods]{Methods}). 

In what follows, we detail how these decompositions leads to efficient classical simulation of circuits with only self-Kerr gates, and we note that the results also apply to cross-Kerr gates. Given a bosonic circuit with $L$ layers of Gaussian and self-Kerr gates
\begin{equation}
    \hat U = \hat G_L \hat \kappa (x_L) \dots \hat G_1 \hat \kappa (x_1),
\end{equation}
by decomposing each of the self-Kerr gates as a sum of phase-shifters, the output state can be written as a superposition of Gaussian states, and the techniques from \cite{Dias2024} can be used for computing (marginal) probabilities for Gaussian measurements of the output state.
This gives an algorithm of classical simulation of circuits composed of Gaussian and Kerr gates which is not restricted to computations over subsystems of larger systems, as in Theorem \ref{theo:Fock_state_simulation}.

While we have focused on classical simulation algorithms of circuits with Gaussian and Kerr gates, we further detail in \hyperref[sec:methods]{Methods}  how phase-shifter decompositions of Kerr gates allows for further classical and quantum simulation through circuit cutting.

\section{Discussion}\label{sec:discussion}
In this work, we have introduced a novel characterization of ``physical'' bosonic states as those with bounded exponential-energy. We have shown that these states admit an exponential improvement in effective dimension with respect to a desired precision of approximation, and further established that output states of universal bosonic quantum circuits have bounded exponential-energy. Taken together, these results demonstrate an exponential enhancement in the effective descriptions of a large class of quantum bosonic states.

We expect these results to have a significant impact on the growing field of bosonic learning theory \cite{mele2024learning,zhao2025complexity,iosue2025highermomenttheory,chen2026optimalgaussianlearning}, as the techniques utilized in this field are heavily dependent on the effective descriptions of the systems. In this direction, it would be interesting to investigate whether our learning result (Theorem~\ref{theo:precision_efficient_learning}) can be combined with the compression Lemma of \cite{mele2024learning} to reduce the exponential dependence on the number of modes $m$ to an exponential dependence on the depth $L$ when the energy-preserving non-Gaussian unitary gates are local. 

It would also be of interest to analyze the implications of this result for learning unitary dynamics. We suspect that efficiency gains can be realized in learning bosonic unitaries if the input states are restricted to those with bounded exponential-energy---a restriction that our results prove is physically well-motivated. 

Our findings significantly reduce the complexity of classical simulation of bosonic systems. While the improvement is not always exponential, due to additional factors beyond the effective dimension determining the complexity of classical simulation, we identified a regime where the reduction in computational complexity becomes exponential. Our results also provide a rigorous justification as well as possible efficiency gains to ad-hoc numerical schemes where finite boson space truncations are employed for the classical simulation of CV systems \cite{blair2025faulttolerant}. Additionally, the introduction of approximate decompositions of Kerr gates as sums of phase-shifters (Lemma~\ref{lem:kerr_decomp}) yields a novel simulation algorithm and enables circuit cutting methods for classical and quantum simulation, as demonstrated in the \hyperref[sec:methods]{Methods} section. 

Finally, we note that the simulation of Hamiltonian dynamics of bosonic systems via qubit-based or hybrid quantum computers \cite{Tong2022provablyaccurate, crane2024hybridoscillator} relies heavily on the effective description of the initial state of the system. In a similar vein, quantum algorithms for state synthesis in photonic platforms also rely on effective descriptions of the target state \cite{Arrazola_2019}. In both cases, our results could improve the efficiency of such quantum algorithms.

Overall, our results show that physically relevant bosonic systems have exponentially smaller effective dimensions than previously assumed, enabling efficient descriptions at very high precision.

\section{Methods}\label{sec:methods}
In this section, we provide the formal versions of Theorems \ref{theo:exp_energy_bound_simplified} and \ref{theo:Fock_state_simulation}, as well as detail the decompositions of Kerr gates as a sum of phase-shifters, and how they lead to new efficient classical simulation algorithms as well as circuit cutting methods.
\subsection{Bound on exponential-energy of output states of universal bosonic circuits}
Here, we formalize Theorem \ref{theo:exp_energy_bound_simplified} from the main text, giving the expression for the upper bound on the exponential-energy for states $\ket{\psi} \in \mathcal C_L$ (Definition \ref{defi:G_nG}), as well as detail how to pick $t_L > 1$:
\begin{theo}[Output states of universal bosonic circuits have bounded exponential-energy]\label{methodtheo:exp_energy_bound_simplified}
Given an $m$-mode quantum state $\ket{\psi} \in \mathcal C_L$ (Definition \ref{defi:G_nG}), with its associated Gaussian unitary gates $\hat G_i$ characterized by displacement vector $\bm{\alpha}_i$ and squeezing vector $\bm{r}_i$, $\forall i \in (1,\dots,L)$, then defining
    \begin{eqnarray}
        |\alpha|&:=& \max_{i,k} |\alpha_{ik}|,\hspace{5mm}
        r := \max_{i,k} r_{ik},
    \end{eqnarray} 
    $\forall i \in (1,\dots,L),k \in (1,\dots,m)$, if we choose $t_L > 1$ according to the scheme
    \begin{equation}\label{eq:choice_t}
        t_i =1 +  \frac{t_{i-1} - 1}{e^{2r} + t_{i-1}}, \forall i \in (1,\dots,L),
    \end{equation}
    with $t_0 = 1 + 2/e^{2r}$, then it follows that
    \begin{eqnarray}\label{eq:exp_energy_bound}
        \bra{\psi}t_L^{\hat N}\ket{\psi} &<& e^{mL^2(|\alpha|^2 + 28 r + 9)} \nonumber \\&<&  t_L^{mL^2 e^{2(2r+1)L} (|\alpha|^2 + 28 r +9)}.
    \end{eqnarray}
\end{theo}
The proof of Theorem \ref{methodtheo:exp_energy_bound_simplified} is given in Section \ref{appsec:exp_energy_bound energy simplified} of the Supplementary Material and relies on tracking the increase in the exponential-energy by displacement and squeezing gate. While the increase induced by displacement gates follows from \cite[Lemma 13]{cerf2025}, in Section \ref{appsec:exp_energy_bound energy simplified} of the Supplementary Material, we prove the following result on the increase in exponential-energy by the squeezing gate. With $s_r\coloneqq\sinh r$:

\begin{lem}[Exponential-energy operator under multimode squeezing] Given an $m$-mode quantum state $\ket{\psi}$ evolving under the $m$-mode squeezing operator $\hat S(\bm{r})$, where $\bm{r} = (r_1,r_2,\dots,r_m) \in \R_+^m, \forall i \in (1,\dots,m)$, and the number operator $\hat N$, then it follows that
\begin{align}
    \bra{\psi}\hat S^\dagger(\bm{r}) t^{\hat{N}}\hat S(\bm{r})\ket{\psi}< g(r,s,t)^m \bra{\psi}s^{\hat N}\ket{\psi},
\end{align}
where
\begin{eqnarray}
    g(r,s,t)&:=& s\times\left(\frac{e^{2r} + s_{r} e^{r}(s-1)}{e^{-2r} - s_{r} e^{-r}(s-1)}\right)^{1/2 } \nonumber \\&\times& \frac{1}{(s-1) - (t-1)(e^{2r} + s_{r} e^{r}(s-1))}, \nonumber \\
\end{eqnarray}
and $r= \max_i r_i ,\forall i \in (1,2,\dots,m)$, provided that
\begin{equation}
    1 < t < 1 + \frac{s-1}{e^{2r} + s_{r} e^{r} (s-1)} < s < \frac{1}{\tanh(r)}.
\end{equation}
\end{lem}
\subsection{Classical simulation of universal bosonic quantum circuits by tracking Fock space amplitudes}
Here we give the formal statement of Theorem \ref{theo:Fock_state_simulation}, where we show how we can efficiently classically simulate a class of bosonic computational problems involving Gaussian and energy-preserving non-Gaussian gate circuits by keeping track of the classical description of the state in the Fock space throughout the computation:
\begin{theo}[Simulating bosonic computations with bounded exponential-energy]
    Let $n$ be a size parameter.  Given an $m$-mode vacuum state $\ket{0}^{\otimes m}$ evolving through the unitary gate
    \begin{equation}
        \hat U = \hat G_L \hat \chi_L \dots \hat G_1 \hat \chi_1,
    \end{equation}
    where $\hat G_1, \dots, \hat G_L$ are arbitrary Gaussian gates, and $\hat \chi_1, \dots, \hat \chi_L$ are energy-preserving non-Gaussian gates such that the action of $\chi_1,\dots,\chi_L$ on Fock states can be computed in time at most exponential in the boson number of the corresponding Fock state. Then, given that $m = \mathcal O(\log(n)), L = \mathcal O(\poly(n))$ and there exists a constant $s > 1$ such that the expectation value of $s^{\hat N}$ is bounded by $\mathcal{O}(\poly(n))$ throughout the computation, then the probabilities generated by $k$-local Gaussian state projections
    \begin{equation}
        P = \left|\bra{G_1\dots G_k}\otimes \mathbb{I}_{m-k} \hat U \ket{0}^{\otimes m}\right|^2
    \end{equation}
    can be computed up to $\mathcal{O}(1/\poly(n))$ additive precision in time $\mathcal O(\poly(n))$, $\forall k \in (1,\dots,m)$, which allows for approximate strong simulation and in particular, approximate sampling.
\end{theo}
\noindent Theorem \ref{theo:exp_energy_bound_simplified} combined with Lemma \ref{lem:efficient_description} guarantee that the effective Fock space throughout the computation is small enough that this algorithm remains efficient. We detail the proof in Section \ref{appsec:Fock_space_simulation} of the Supplementary Material.
\subsection{Phase-shifter decompositions of self-Kerr and cross-Kerr gates}\label{methodsec:phase_shifter_simul}
In this section, we introduce gate decompositions for Kerr gates. These decompositions were obtained in \cite{Tara1993,Jun-wei1996} for self-Kerr gates, and we give a generalisation to cross-Kerr gates:
\begin{lem}[Decomposition of Kerr/cross-Kerr gate as phase shifters]\label{lem:kerr_decomp}
    Given integers $p\in \Z, q \in \Z_+^*$, the self-Kerr gate $\hat \kappa (- p/q)$ can be expressed as
    \begin{equation}
    \hat \kappa\left(- \frac{p}{q}\right) = \sum_{j=0}^{q-1} g_{j,p,q}^{(e)} \exp\left(-i \frac{2\pi j}{q} \hat N_1\right)
\end{equation}
if $q$ is even, and as
\begin{equation}
    \hat\kappa\left(-\frac{p}{q}\right) = \sum_{j=0}^{q-1} g_{j,p,q}^{(o)} \exp\left(-i \frac{\pi}{q} (2j+p) \hat N_1\right)
\end{equation}
if $q$ is odd. Furthermore, for all integers $p\in \Z, q \in \Z_+^*$, the cross-Kerr gate can be decomposed as 
\begin{equation}
    \hat{c\kappa} \left(\frac{p}{q}\right) = \sum_{j,k = 0}^{q-1} g_{j,k,p,q} \exp\left(i \frac{2\pi}{q}(j \hat N_1 + k \hat N_2)\right).
\end{equation}
The expressions for $g_{j,p,q}^{(e)}, g_{j,p,q}^{(o)}$ and $g_{j,k,p,q}$ are given in the Supplementary Material.
\end{lem}
\noindent The proof of Lemma \ref{lem:kerr_decomp} is given in Section \ref{appsec:kerr_decomp} of the Supplementary Material, and follows from analysing the periodicity of the Kerr gates with respect to $q$, which allows us to write the Kerr gates as sums of phase-shifters that are periodic in $q$. The coefficients then follow from the inverse Fourier transform of the expressions. We expect a similar decomposition to hold for non-Gaussian gates of the form $\exp(iy\hat N^k)$ for $k>2$, as long as $y$ has the right periodicity. 

Now we detail the application of Lemma \ref{lem:kerr_decomp} in efficient classical simulation and circuit cutting in circuits composed of Gaussian and Kerr gates. For simplicity, we give our result for Gaussian and self-Kerr gates and note that the analysis extends to circuits including cross-Kerr gates as well.
\subsection{Classical simulation of Gaussian and Kerr circuits}
Given an $m$-mode vacuum state evolving through a circuit described by the unitary gate
\begin{equation}
    \hat U = \hat G_L \hat \kappa (x_L) \dots \hat G_1 \hat \kappa (x_1),
\end{equation}
if all the self-Kerr gates have rational parameters, i.e.
\begin{equation}
    x_i = p_i/q_i, p_i \in \Z, q_i \in \Z_+,
\end{equation}
then by Lemma \ref{lem:kerr_decomp}, the self-Kerr gate $\hat \kappa (x_i)$ can be written as a sum of $q_i$ phase-shifters, and the output state can be written as a sum of Gaussian states. The probabilities generated by Gaussian measurements on the output state can then be computed using techniques from \cite{Dias2024}.

However, if some of the self-Kerr gates do not have rational parameters, i.e.\ their parameters $x$ cannot be put in the form $ p/q$, we use a Diophantine approximation \cite{hurwitz1891ueber}: given an irrational number $x$, we can use continued fractions to find a rational number $p/q$ such that
\begin{equation}
    \left|x- \frac{p}{q}\right| < \frac{1}{\sqrt{5}q^2}.
\end{equation}
The trace distance error induced by this approximation is quantified by the following Lemma, which leverages an upper bound on the exponential-energy:
\begin{lem}[Error induced by ignoring a self-Kerr gate of small strength]\label{lem:eps_kerr_gate_td}
    Given a multimode state $\ket{\psi}$ such that  $\bra{\psi} e^{-i\epsilon \hat N_1^2} s^{\hat N} e^{i\epsilon \hat N_1^2} \ket{\psi} = \bra{\psi}  s^{\hat N}\ket{\psi} \leq s^E$, the trace distance between $\ket{\psi}$ and $e^{i\epsilon \hat N_1^2}\ket{\psi}$ is upper bounded by $\mathcal O(\epsilon E^2 \log_s^2 (1/\epsilon^2))$.
\end{lem}
\noindent The proof of Lemma \ref{lem:eps_kerr_gate_td} is given in Section \ref{appsec:eps_kerr_gate_td} of the Supplementary Material and follows by first truncating $\ket{\psi}$ and $e^{i\epsilon \hat N_1^2} \ket{\psi}$ to the Hilbert space with total bosons less than equal to $k$ (with truncation error quantified by Lemma \ref{lem:efficient_description}) and then calculating the trace distance between the truncated states. We optimize the cutoff boson number $k$ to find the upper bound.

From Theorem \ref{theo:exp_energy_bound_simplified} and the subsequent discussion, we know that for a circuit composed of Gaussian and self-Kerr gates, after each layer of Gaussian and self-Kerr gate, we can find $t_i > 1$ and $E_i \in \R_+$ such that $\bra{\psi_i}t_i^{\hat N}\ket{\psi_i} < t_i^{E_i}$, where $E_i$ can be interpreted as energy of the circuit. Therefore, for a circuit composed of $L$ layers of Gaussian and self-Kerr gates (with possibly irrational parameters), we first approximate the irrational Kerr gates with rational gates using the Diophantine approximation where the precision with which we approximate the irrational parameter is determined according to Lemma \ref{lem:eps_kerr_gate_td}. This allows us to approximate the output state of the Gaussian and self-Kerr gate circuit as a superposition of Gaussian states, and we can again use the techniques from \cite{Dias2024} to approximately compute the output probabilities. This gives the following result on simulation of Gaussian and self-Kerr gate circuits:
\begin{theo}[Simulating universal bosonic computations with self-Kerr gates]
         Given the $m$-mode vacuum state $\ket{0}^{\otimes m}$ evolving through a quantum circuit characterized by the unitary 
    \begin{equation}
        \hat U = \hat G_L \hat \kappa(x_L) \dots \hat G_1 \hat \kappa(x_1),
    \end{equation}
    where $\hat G_1,\dots,\hat G_L$ are arbitrary Gaussian gates, $\hat \kappa (x_1),\dots,\hat \kappa (x_L)$ are self-Kerr gates, then for any $k \in (1,\dots,m)$, probabilities for heterodyne measurements on the first $k$ modes of the output state
    \begin{equation}
        P (\bm\alpha_k) = \frac{1}{\pi^k}\left|\bra{\alpha_1 \alpha_2 \dots \alpha_k}\otimes \mathbb{I}_{m-k}\hat U\ket{0}^{\otimes m}\right|^2
    \end{equation}
    \begin{itemize}
        \item can be computed exactly in polynomial time for log depth circuits $L = \mathcal O(\log(\poly(m)))$ if all the self-Kerr gates have rational parameters, i.e.
        \begin{equation}
            x_i =  p_i/q_i, \; p_i \in \Z,q_i \in \Z_+^*, \;\forall i \in (1,\dots,L);
        \end{equation}
        \item can be computed up to additive precision $\mathcal O(1/\poly(m))$ in quasipolynomial time for log depth circuits $L = \mathcal O(\log(\poly(m)))$ for generic self-Kerr gates.
    \end{itemize}
\end{theo}
\noindent Note that here we are considering heterodyne measurements without loss of generality, since arbitrary single-mode Gaussian measurements can be written as a single-mode Gaussian gate followed by a heterodyne measurement \cite{chabaud2021corestates,giedke2002characterization}. We can then absorb the single-mode Gaussian gates into the final Gaussian gate $\hat G_L$.
\subsection{Circuit cutting for bosonic circuits composed of Gaussian and Kerr gates}

Quantum circuit cutting refers to a method to simulate a large quantum computation by combining simpler quantum computations together with classical post-processing (see \cite{barral2025review} for a review of circuit cutting in the context of qubit-based computations). At a high level, this method partitions an initial circuit into simpler chunks by decomposing a quantum channel as a sum of simpler channels via a so-called quasi-probability decomposition.
While circuit-cutting is typically used for replacing entangling gates by sums of separable ones, the decomposition of Kerr gates as sums of phase-shifters (Lemma \ref{lem:kerr_decomp}) allows us to derive a similar circuit cutting method to simulate bosonic circuits composed of Gaussian and Kerr gates using circuits with fewer Kerr gates, and even Gaussian circuits. 

Concretely, we write the circuit described by the unitary 
\begin{equation}
    \hat U = \hat G_L \hat \kappa (p_l/q_l) \dots \hat G_1 \hat \kappa (p_1/q_1)
\end{equation}
(for simplicity, we are assuming that the self-Kerr gates are all rational, using the Diophantine approximation, we can always make sure this is the case and quantify the error using Lemma \ref{lem:eps_kerr_gate_td}) as a sum of Gaussian unitary gates 
\begin{equation}
    \hat U = \sum_{j_1=0}^{q_1 - 1}\dots\sum_{j_L=0}^{q_L - 1} g_{j_1,p_1,q_1}^{(e/o)} \dots g_{j_1,p_1,q_1}^{(e/o)} \hat G_{j_1,\dots,j_L}.
\end{equation}
Therefore, given the $m$-mode vacuum state as the input state, the probabilities generated by global Gaussian projections on the output state,
\begin{equation}
    P(\bm G) = \left|\bra{G_1\dots G_m}\hat U\ket{0}^{\otimes m}\right|^2,
\end{equation}
can be estimated as follows. We have
\begin{eqnarray}\label{eq:sampled_expression}
    P(\bm G)  &=& \sum_{J,J'} g_J^{(e/o)}g_{J'}^{(e/o)} \braket{\bm G| G_J} \braket{ G_{J'}|\bm G} \nonumber \\
&=& \|g\|_1 \sum_{J,J'} \frac{g_J^{(e/o)}g_{J'}^{(e/o)} \braket{\bm G| G_J} \braket{ G_{J'}|\bm G}}{\|g\|_1}, \nonumber \\
\end{eqnarray}
where in the first line, we have defined $J:= (j_1,\dots, j_L)$, $J'= (j_1',\dots,j_L')$, $g_J^{(e/o)} = \prod_{i=1}^L g_{j_i,p_1,q_1}^{(e/o)}$ and similarly for $g_J$. Note that each of the $j_i (j'_i)$ runs from 0 to $q_i - 1$. In the second line, we have defined
\begin{eqnarray} 
   \|g\|_1 &=& \sum_{J,J'} |g_{J}^{(e/o)} \dots  g_{J'}^{(e/o)}| \nonumber \\ 
   &=& \sum_{j_1=0}^{q_1 - 1} |g_{j_1,p_1,q_1}^{(e/o)}|\dots \sum_{j_L'=0}^{q_L - 1} |g_{j_L',p_L,q_L}^{(e/o)}| \nonumber \\
   &=& \prod_{i = 1}^L \|g_{q_i}\|_1.
\end{eqnarray}
We estimate $|\bra{G_1\dots G_m}\hat U\ket{0}^{\otimes m}|^2$ by statistically sampling from the expression in Eq.~\ref{eq:sampled_expression} $N$ times and calculating the two Gaussian overlaps at the sampled points either classically or by building the appropriate Gaussian circuit (that are comparatively easier to implement) and averaging the value. Let us call the obtained statistical average $\mu$. Then by Hoeffding's inequality, it follows that
\begin{eqnarray}
    &&\mathrm{Pr}\left[\left|P(\bm G) - \mu\right| \geq \epsilon\right] \nonumber \leq 2\exp\left(-\frac{N\epsilon^2}{2\|g\|_1^2 }\right).
\end{eqnarray}
Therefore, for a success probability $1-\delta$, it is sufficient to take a number of samples
\begin{equation}
    N =\frac{2\|g\|_1^2}{\epsilon^2}\log\left(\frac2\delta\right).
\end{equation}
As we prove in Section \ref{appsec:two-norm_bounds} of the Supplementary Material, $\|g_{q_i}\|_1 \leq \sqrt{q_i}$. Therefore, $\|g\|_1^2 \leq q^{2L}$, where $q = \max_i q_i$ and we obtain a  sample complexity of probability estimation up to additive precision $\epsilon$ with success probability $1-\delta$ by circuit cutting to be
\begin{equation}
    N = \frac{2q^{2L}}{\epsilon^2}\log\left(\frac2\delta\right).
\end{equation}
The same approach also allows us to simulate bosonic circuits composed of Gaussian and Kerr gates using circuits with fewer Kerr gates by decomposing a subset of the Kerr gates as sums of Gaussian gates.


\medskip

\noindent\textit{Note added.} During the preparation of this work, we became aware of the independent work \cite{Guseynov2026} which derives classical simulation algorithms for bosonic circuits composed of Gaussian and Kerr gates, using coherent state propagation methods.

\section{Acknowledgements}
V.U.\ thanks S.\ A.\ Sater and H.\ Thomas for insightful discussions. N.Q.\ thanks A.J.\ Boon for pointers to valuable references. U.C.\ thanks D.\ Rudolph, S.\ Mehraban, M.\ Fanizza and C.\ Rouzé for interesting discussions. We acknowledge support from the PSL Global Seed Fund 2025. U.C.\ and V.U.\ acknowledge funding from the European Union’s Horizon Europe Framework Programme (EIC Pathfinder Challenge project Veriqub) under Grant Agreement No.~101114899.

\bibliographystyle{linksen}
\bibliography{biblio}

\providecommand{\href}[2]{#2}\begingroup\raggedright\begin{thebibliography}{10}

\bibitem{bose1924plancks}
S.~N. Bose, ``Plancks Gesetz und Lichtquantenhypothese,'' \href{http://dx.doi.org/https://doi.org/10.1007/BF01327326}{{\em Zeitschrift f{\"u}r Physik} {\bfseries 26}, 178--181 (1924)}.

\bibitem{London1938}
F.~London, ``The $\lambda$-Phenomenon of Liquid Helium and the Bose-Einstein Degeneracy,'' \href{http://dx.doi.org/10.1038/141643a0}{{\em Nature} {\bfseries 141}, 643--644 (1938)}. \url{https://doi.org/10.1038/141643a0}.

\bibitem{Hong1987}
C.~K. Hong, Z.~Y. Ou, and L.~Mandel, ``Measurement of subpicosecond time intervals between two photons by interference,'' \href{http://dx.doi.org/10.1103/PhysRevLett.59.2044}{{\em Phys. Rev. Lett.} {\bfseries 59}, 2044--2046 (1987)}.

\bibitem{Gottesman2001}
D.~Gottesman, A.~Kitaev, and J.~Preskill, ``Encoding a qubit in an oscillator,'' \href{http://dx.doi.org/10.1103/PhysRevA.64.012310}{{\em Phys. Rev. A} {\bfseries 64}, 012310 (2001)}.

\bibitem{Knill2001}
E.~Knill, R.~Laflamme, and G.~Milburn, ``A scheme for efficient quantum computation with linear optics,'' \href{http://dx.doi.org/10.1038/35051009}{{\em Nature} {\bfseries 409}, 46--52 (2001)}.

\bibitem{Menicucci2006}
N.~C. Menicucci, P.~van Loock, M.~Gu, C.~Weedbrook, T.~C. Ralph, and M.~A. Nielsen, ``Universal Quantum Computation with Continuous-Variable Cluster States,'' \href{http://dx.doi.org/10.1103/PhysRevLett.97.110501}{{\em Phys. Rev. Lett.} {\bfseries 97}, 110501 (2006)}.

\bibitem{madsen2022quantum}
L.~S. Madsen, F.~Laudenbach, M.~F. Askarani, F.~Rortais, T.~Vincent, J.~F. Bulmer, F.~M. Miatto, L.~Neuhaus, L.~G. Helt, M.~J. Collins, {\em et al.}, ``Quantum computational advantage with a programmable photonic processor,'' \href{http://dx.doi.org/10.1038/s41586-022-04725-x}{{\em Nature} {\bfseries 606}, 75--81 (2022)}.

\bibitem{yokoyama2013ultra}
S.~Yokoyama, R.~Ukai, S.~C. Armstrong, C.~Sornphiphatphong, T.~Kaji, S.~Suzuki, J.-i. Yoshikawa, H.~Yonezawa, N.~C. Menicucci, and A.~Furusawa, ``Ultra-large-scale continuous-variable cluster states multiplexed in the time domain,'' \href{http://dx.doi.org/10.1038/nphoton.2013.287}{{\em Nature Photonics} {\bfseries 7}, 982 (2013)}.

\bibitem{larsen2025integrated}
M.~Larsen, J.~Bourassa, S.~Kocsis, J.~Tasker, R.~Chadwick, C.~Gonz{\'a}lez-Arciniegas, J.~Hastrup, C.~Lopetegui-Gonz{\'a}lez, F.~Miatto, A.~Motamedi, {\em et al.}, ``Integrated photonic source of Gottesman--Kitaev--Preskill qubits,'' \href{http://dx.doi.org/https://doi.org/10.1038/s41586-025-09044-5}{{\em Nature} 1--5 (2025)}.

\bibitem{michael2016new}
M.~H. Michael, M.~Silveri, R.~Brierley, V.~V. Albert, J.~Salmilehto, L.~Jiang, and S.~M. Girvin, ``New class of quantum error-correcting codes for a bosonic mode,'' \href{http://dx.doi.org/10.1103/PhysRevX.6.031006}{{\em Physical Review X} {\bfseries 6}, 031006 (2016)}.

\bibitem{terhal2020towards}
B.~M. Terhal, J.~Conrad, and C.~Vuillot, ``Towards scalable bosonic quantum error correction,'' \href{http://dx.doi.org/10.1088/2058-9565/ab98a5}{{\em Quantum Science and Technology} {\bfseries 5}, 043001 (2020)}.

\bibitem{Yalouz2021encodingstrongly}
S.~Yalouz, B.~Senjean, F.~Miatto, and V.~Dunjko, ``Encoding strongly-correlated many-boson wavefunctions on a photonic quantum computer: application to the attractive {B}ose-{H}ubbard model,'' \href{http://dx.doi.org/10.22331/q-2021-11-08-572}{{\em {Quantum}} {\bfseries 5}, 572 (2021)}.

\bibitem{Tong2022provablyaccurate}
Y.~Tong, V.~V. Albert, J.~R. McClean, J.~Preskill, and Y.~Su, ``Provably accurate simulation of gauge theories and bosonic systems,'' \href{http://dx.doi.org/10.22331/q-2022-09-22-816}{{\em {Quantum}} {\bfseries 6}, 816 (2022)}.

\bibitem{hanada2025quantumbosons}
M.~Hanada, S.~Matsuura, E.~Mendicelli, and E.~Rinaldi, ``Exponential improvement in quantum simulations of bosons,'' \href{http://arxiv.org/abs/2505.02553}{{\ttfamily arXiv:2505.02553 [quant-ph]}}.

\bibitem{Mari2012}
A.~Mari and J.~Eisert, ``Positive {W}igner Functions Render Classical Simulation of Quantum Computation Efficient,'' \href{http://dx.doi.org/10.1103/PhysRevLett.109.230503}{{\em Physical Review Letters} {\bfseries 109}, 230503 (2012)}.

\bibitem{chabaud2023resources}
U.~Chabaud and M.~Walschaers, ``Resources for bosonic quantum computational advantage,'' \href{http://dx.doi.org/10.1103/PhysRevLett.130.090602}{{\em Physical Review Letters} {\bfseries 130}, 090602 (2023)}.

\bibitem{Dias2024}
B.~Dias and R.~K\"onig, ``Classical simulation of non-Gaussian bosonic circuits,'' \href{http://dx.doi.org/10.1103/PhysRevA.110.042402}{{\em Phys. Rev. A} {\bfseries 110}, 042402 (2024)}.

\bibitem{calcluth2025}
C.~Calcluth, O.~Hahn, J.~Bermejo-Vega, A.~Ferraro, and G.~Ferrini, ``Classical Simulation of Circuits with Realistic Odd-Dimensional Gottesman-Kitaev-Preskill States,'' \href{http://dx.doi.org/10.1103/xmtw-g54f}{{\em Phys. Rev. Lett.} {\bfseries 135}, 010601 (2025)}.

\bibitem{upreti2025interplay}
V.~Upreti and U.~Chabaud, ``Interplay of resources for universal continuous-variable quantum computing,'' \href{http://dx.doi.org/https://doi.org/10.48550/arXiv.2502.07670}{{\em arXiv preprint arXiv:2502.07670} (2025)}.

\bibitem{arzani2025}
F.~Arzani, R.~I. Booth, and U.~Chabaud, ``Effective descriptions of bosonic systems can be considered complete,'' \href{http://dx.doi.org/10.1038/s41467-025-64872-3}{{\em Nature Communications} {\bfseries 16}, 9744 (2025)}.

\bibitem{mele2024learning}
F.~A. Mele, A.~A. Mele, L.~Bittel, J.~Eisert, V.~Giovannetti, L.~Lami, L.~Leone, and S.~F. Oliviero, ``Learning quantum states of continuous-variable systems,'' \href{http://dx.doi.org/https://doi.org/10.1038/s41567-025-03086-2}{{\em Nature Physics} 1--7 (2025)}.

\bibitem{mele2025communication}
F.~A. Mele, G.~Barbarino, V.~Giovannetti, and M.~Fanizza, ``Achievable rates in non-asymptotic bosonic quantum communication,'' \href{http://arxiv.org/abs/2502.05524}{{\ttfamily arXiv:2502.05524 [quant-ph]}}.

\bibitem{cerf2025}
S.~Cerf, C.~Wassner, J.~Davis, F.~Arzani, and U.~Chabaud, ``On the complex zeros of the wavefunction,'' \href{http://arxiv.org/abs/2507.23468}{{\ttfamily arXiv:2507.23468 [quant-ph]}}.

\bibitem{Lloyd1999}
S.~Lloyd and S.~L. Braunstein, ``Quantum Computation over Continuous Variables,'' \href{http://dx.doi.org/10.1103/PhysRevLett.82.1784}{{\em Phys. Rev. Lett.} {\bfseries 82}, 1784--1787 (1999)}.

\bibitem{NielsenChuang}
M.~A. Nielsen and I.~L. Chuang, ``Quantum Computation and Quantum Information: 10th Anniversary Edition,''.
\newblock \href{http://dx.doi.org/10.1017/CBO9780511976667}{Cambridge University Press}, New York, NY, USA, 2011.

\bibitem{Braunstein2005}
S.~L. Braunstein and P.~van Loock, ``Quantum information with continuous variables,'' \href{http://dx.doi.org/10.1103/RevModPhys.77.513}{{\em Rev. Mod. Phys.} {\bfseries 77}, 513--577 (2005)}.

\bibitem{ferraro2005}
A.~Ferraro, S.~Olivares, and M.~G.~A. Paris, ``Gaussian states in continuous variable quantum information,'' \href{http://arxiv.org/abs/quant-ph/0503237}{{\ttfamily arXiv:quant-ph/0503237 [quant-ph]}}.

\bibitem{Weedbrook2012}
C.~Weedbrook, S.~Pirandola, R.~Garc\'{\i}a-Patr\'on, N.~J. Cerf, T.~C. Ralph, J.~H. Shapiro, and S.~Lloyd, ``Gaussian quantum information,'' \href{http://dx.doi.org/10.1103/RevModPhys.84.621}{{\em Rev. Mod. Phys.} {\bfseries 84}, 621--669 (2012)}.

\bibitem{Winter1999}
A.~Winter, ``Coding theorem and strong converse for quantum channels,'' \href{http://dx.doi.org/10.1109/18.796385}{{\em IEEE Transactions on Information Theory} {\bfseries 45}, 2481--2485 (1999)}.

\bibitem{AnnaMele2024}
F.~A. Mele, A.~A. Mele, L.~Bittel, J.~Eisert, V.~Giovannetti, L.~Lami, L.~Leone, and S.~F.~E. Oliviero, ``Learning quantum states of continuous variable systems,'' 2024.
\newblock \url{https://arxiv.org/abs/2405.01431}.

\bibitem{BUCCO}
A.~Serafini, ``Quantum Continuous Variables: A Primer of Theoretical Methods,''.
\newblock \href{http://dx.doi.org/10.1201/9781315118727}{CRC Press, Taylor \& Francis Group}, Boca Raton, USA, 2017.

\bibitem{Sefi2011}
S.~Sefi and P.~van Loock, ``How to Decompose Arbitrary Continuous-Variable Quantum Operations,'' \href{http://dx.doi.org/10.1103/PhysRevLett.107.170501}{{\em Phys. Rev. Lett.} {\bfseries 107}, 170501 (2011)}.

\bibitem{Marshall2015}
K.~Marshall, R.~Pooser, G.~Siopsis, and C.~Weedbrook, ``Repeat-until-success cubic phase gate for universal continuous-variable quantum computation,'' \href{http://dx.doi.org/10.1103/PhysRevA.91.032321}{{\em Phys. Rev. A} {\bfseries 91}, 032321 (2015)}.

\bibitem{Miyata2016}
K.~Miyata, H.~Ogawa, P.~Marek, R.~Filip, H.~Yonezawa, J.-i. Yoshikawa, and A.~Furusawa, ``Implementation of a quantum cubic gate by an adaptive non-Gaussian measurement,'' \href{http://dx.doi.org/10.1103/PhysRevA.93.022301}{{\em Phys. Rev. A} {\bfseries 93}, 022301 (2016)}.

\bibitem{Alexander2018}
R.~N. Alexander, S.~Yokoyama, A.~Furusawa, and N.~C. Menicucci, ``Universal quantum computation with temporal-mode bilayer square lattices,'' \href{http://dx.doi.org/10.1103/PhysRevA.97.032302}{{\em Phys. Rev. A} {\bfseries 97}, 032302 (2018)}.

\bibitem{zheng2023gaussian}
Y.~Zheng, A.~Ferraro, A.~F. Kockum, and G.~Ferrini, ``Gaussian conversion protocol for heralded generation of generalized Gottesman-Kitaev-Preskill states,'' \href{http://dx.doi.org/10.1103/PhysRevA.108.012603}{{\em Phys. Rev. A} {\bfseries 108}, 012603 (2023)}.

\bibitem{Budinger2024}
N.~Budinger, A.~Furusawa, and P.~van Loock, ``All-optical quantum computing using cubic phase gates,'' \href{http://dx.doi.org/10.1103/PhysRevResearch.6.023332}{{\em Phys. Rev. Res.} {\bfseries 6}, 023332 (2024)}.

\bibitem{chabaud2025energy}
U.~Chabaud, S.~Gharibian, S.~Mehraban, A.~Motamedi, H.~R. Naeij, D.~Rudolph, and D.~Sambrani, ``Energy, Bosons and Computational Complexity,'' \href{http://arxiv.org/abs/2510.08545}{{\ttfamily arXiv:2510.08545 [quant-ph]}}.

\bibitem{hastrup2021unsuitability}
J.~Hastrup, M.~V. Larsen, J.~S. Neergaard-Nielsen, N.~C. Menicucci, and U.~L. Andersen, ``Unsuitability of cubic phase gates for non-Clifford operations on Gottesman-Kitaev-Preskill states,'' \href{http://dx.doi.org/10.1103/PhysRevA.103.032409}{{\em Physical Review A} {\bfseries 103}, 032409 (2021)}.

\bibitem{boudreault2026using}
J.~Boudreault, R.~Shillito, J.-B. Bertrand, and B.~Royer, ``Using a Self-Kerr Nonlinearity for Magic State Preparation in Grid Codes,'' {\em Physical Review Letters} {\bfseries 136}, 120601 (2026).

\bibitem{hall2013quantum}
B.~C. Hall, ``Quantum theory for mathematicians,'', vol.~267.
\newblock \href{http://dx.doi.org/10.1007/978-1-4614-7116-5}{Springer}, 2013.

\bibitem{anshu2023surveycomplexity}
A.~Anshu and S.~Arunachalam, ``A survey on the complexity of learning quantum states,'' \href{http://dx.doi.org/10.1038/s42254-023-00662-4}{{\em Nature Reviews Physics} {\bfseries 6}, 59--69 (2024)}.

\bibitem{chabaud2021corestates}
U.~Chabaud, G.~Ferrini, F.~Grosshans, and D.~Markham, ``Classical simulation of Gaussian quantum circuits with non-Gaussian input states,'' \href{http://dx.doi.org/10.1103/PhysRevResearch.3.033018}{{\em Phys. Rev. Res.} {\bfseries 3}, 033018 (2021)}.

\bibitem{Tara1993}
K.~Tara, G.~S. Agarwal, and S.~Chaturvedi, ``Production of Schr\"odinger macroscopic quantum-superposition states in a Kerr medium,'' \href{http://dx.doi.org/10.1103/PhysRevA.47.5024}{{\em Phys. Rev. A} {\bfseries 47}, 5024--5029 (1993)}.

\bibitem{Jun-wei1996}
W.~Jin-wei and G.~Guang-can, ``Superposition of coherent states in Kerr medium,'' \href{http://dx.doi.org/10.1088/1004-423X/5/10/006}{{\em Acta Physica Sinica (Overseas Edition)} {\bfseries 5}, 764 (1996)}.

\bibitem{hurwitz1891ueber}
A.~Hurwitz, ``Ueber die angenäherte Darstellung der Irrationalzahlen durch rationale Brüche,'' \href{http://dx.doi.org/10.1007/BF01206656}{{\em Mathematische Annalen} {\bfseries 39}, 279--284 (1891)}.

\bibitem{zhao2025complexity}
X.~Zhao, P.~Liao, F.~A. Mele, U.~Chabaud, and Q.~Zhuang, ``Complexity of quantum tomography from genuine non-Gaussian entanglement,'' \href{http://dx.doi.org/https://doi.org/10.1038/s41467-025-67062-3}{{\em Nature Communications} (2025)}.

\bibitem{iosue2025highermomenttheory}
J.~T. Iosue, Y.-X. Wang, I.~Datta, S.~Ghosh, C.~Oh, B.~Fefferman, and A.~V. Gorshkov, ``Higher moment theory and learnability of bosonic states,'' \href{http://arxiv.org/abs/2510.01610}{{\ttfamily arXiv:2510.01610 [quant-ph]}}.

\bibitem{chen2026optimalgaussianlearning}
S.~Chen, F.~A. Mele, M.~Fanizza, A.~Li, Z.~Mann, H.-Y. Huang, Y.~Chen, and J.~Preskill, ``Towards sample-optimal learning of bosonic Gaussian quantum states,'' \href{http://arxiv.org/abs/2603.18136}{{\ttfamily arXiv:2603.18136 [quant-ph]}}.

\bibitem{blair2025faulttolerant}
S.~Blair, F.~Arzani, G.~Ferrini, and A.~Ferraro, ``Towards fault-tolerant quantum computation with universal continuous-variable gates,'' \href{http://arxiv.org/abs/2506.13643}{{\ttfamily arXiv:2506.13643 [quant-ph]}}.

\bibitem{crane2024hybridoscillator}
E.~Crane, K.~C. Smith, T.~Tomesh, A.~Eickbusch, J.~M. Martyn, S.~Kühn, L.~Funcke, M.~A. DeMarco, I.~L. Chuang, N.~Wiebe, A.~Schuckert, and S.~M. Girvin, ``Hybrid Oscillator-Qubit Quantum Processors: Simulating Fermions, Bosons, and Gauge Fields,'' \href{http://arxiv.org/abs/2409.03747}{{\ttfamily arXiv:2409.03747 [quant-ph]}}.

\bibitem{Arrazola_2019}
J.~M. Arrazola, T.~R. Bromley, J.~Izaac, C.~R. Myers, K.~Br{\'{a}}dler, and N.~Killoran, ``Machine learning method for state preparation and gate synthesis on photonic quantum computers,'' \href{http://dx.doi.org/10.1088/2058-9565/aaf59e}{{\em Quantum Science and Technology} {\bfseries 4}, 024004 (2019)}.

\bibitem{giedke2002characterization}
G.~Giedke and J.~I. Cirac, ``Characterization of {G}aussian operations and distillation of {G}aussian states,'' \href{http://dx.doi.org/10.1103/PhysRevA.66.032316}{{\em Physical Review A} {\bfseries 66}, 032316 (2002)}.

\bibitem{barral2025review}
D.~Barral, F.~J. Cardama, G.~D{\'\i}az-Camacho, D.~Fa{\'\i}lde, I.~F. Llovo, M.~Mussa-Juane, J.~V{\'a}zquez-P{\'e}rez, J.~Villasuso, C.~Pi{\~n}eiro, N.~Costas, {\em et al.}, ``Review of distributed quantum computing: From single qpu to high performance quantum computing,'' \href{http://dx.doi.org/10.1016/j.cosrev.2025.100747}{{\em Computer Science Review} {\bfseries 57}, 100747 (2025)}.

\bibitem{Guseynov2026}
N.~Guseynov, Z.~Holmes, and A.~Angrisani, ``Coherent-State Propagation: A Computational Framework for Simulating Bosonic Quantum Systems,'' \href{http://arxiv.org/abs/2604.xxxxx}{{\ttfamily arXiv:2604.xxxxx [quant-ph]}}.

\bibitem{Kim1989}
M.~S. Kim, F.~A.~M. de~Oliveira, and P.~L. Knight, ``Properties of squeezed number states and squeezed thermal states,'' \href{http://dx.doi.org/10.1103/PhysRevA.40.2494}{{\em Phys. Rev. A} {\bfseries 40}, 2494--2503 (1989)}.

\bibitem{wright2016learn}
J.~Wright, {\em How to learn a quantum state}.
\newblock PhD thesis, Carnegie Mellon University, 2016.

\bibitem{yao2024}
Y.~Yao, F.~Miatto, and N.~Quesada, ``Riemannian optimization of photonic quantum circuits in phase and Fock space,'' \href{http://dx.doi.org/10.21468/SciPostPhys.17.3.082}{{\em SciPost Phys.} {\bfseries 17}, 082 (2024)}.

\bibitem{bjorklund2019fasthafnian}
A.~Björklund, B.~Gupt, and N.~Quesada, ``A faster hafnian formula for complex matrices and its benchmarking on a supercomputer,'' \href{http://arxiv.org/abs/1805.12498}{{\ttfamily arXiv:1805.12498 [cs.DS]}}.

\bibitem{Borwein1985}
P.~B. Borwein, ``On the complexity of calculating factorials,'' \href{http://dx.doi.org/https://doi.org/10.1016/0196-6774(85)90006-9}{{\em Journal of Algorithms} {\bfseries 6}, 376--380 (1985)}.

\bibitem{Reck1994}
M.~Reck, A.~Zeilinger, H.~J. Bernstein, and P.~Bertani, ``Experimental realization of any discrete unitary operator,'' \href{http://dx.doi.org/10.1103/PhysRevLett.73.58}{{\em Phys. Rev. Lett.} {\bfseries 73}, 58--61 (1994)}.

\end{thebibliography}\endgroup
\onecolumn\newpage
\appendix

\begin{center}
    {\Huge Supplementary Material}
\end{center}

\section{Additional Notations}
For the proofs given in this Supplementary Material, we use the following notations for the hyperbolic trigonometric functions for brevity:
\begin{equation}
    s_r := \sinh(r); \hspace{5mm} c_r := \cosh(r); \hspace{5mm} t_r := \tanh(r).
\end{equation}
\section{Effective descriptions of bosonic quantum states with exponential-energy operator bound}\label{appsec:efficient_description}
In this section, we find the boson number cutoff for bosonic quantum states with bounded exponential-energy operator expectation values in terms of the desired precision, as formalized in the following Lemma:
\begin{lem}\label{applem:efficient_description}
Given a quantum state $\ket{\psi} \in \mathcal S_{s,E}$, the normalized Fock state truncation of $\ket{\psi}$ up to total boson number $k$, $\ket{\psi_k} =\frac{\Pi_k \ket{\psi}}{\sqrt{\Tr[\Pi_k \ket{\psi}\bra{\psi}]}}$, is such that
    \begin{equation}
        \frac12 \|\ket{\psi}\bra{\psi} - \ket{\psi_k}\bra{\psi_k}\|_1 \leq \sqrt{\frac{s^E}{s^k}}.
    \end{equation}
    Therefore, choosing
    \begin{equation}\label{appeq:boson_number_cutoff}
        k =E + 2\log_s(1/\epsilon),
    \end{equation} 
    we get that
    \begin{equation}
        \frac12 \|\ket{\psi}\bra{\psi} - \ket{\psi_k}\bra{\psi_k}\|_1 \leq \epsilon.
    \end{equation}
\end{lem}
\begin{proof}
Writing
\begin{eqnarray}
    \ket{\psi} &=& \sum_{\bm n} a_{\bm n} \ket{\bm{n}} ,\nonumber \\
    \ket{\tilde \psi_k} &=& \sum_{|\bm{n}|\leq k} a_{\bm n} \ket{\bm{n}} .
\end{eqnarray}
Further, noting that 
\begin{equation}
    \braket{\tilde\psi_k|\tilde \psi_k} = \sum_{|\bm n |\leq k} |a_{\bm n}|^2 = \braket{\tilde\psi_k|\psi},
\end{equation}
we write
\begin{equation}
    \ket{\psi_k} = \frac{1}{\sqrt{\braket{\psi|\tilde \psi_k}}} \sum_{|\bm{n}|\leq k} a_{\bm n} \ket{\bm{n}} = \frac{\ket{\tilde \psi_k}}{\sqrt{\braket{\psi|\tilde \psi_k}}}.
\end{equation}
Also, note that $\braket{\tilde\psi_k|\psi}$ is real and positive. Now, the trace distance between $\ket{\psi}$ and $\ket{\psi_k}$ can be expressed as
\begin{eqnarray}\label{appeq:eff_descr_a}
     \frac12 \|\ket{\psi}\bra{\psi} - \ket{\psi_k}\bra{\psi_k}\|_1 &=& \sqrt{1 - \frac{\braket{\tilde \psi_k|\psi}^2}{\braket{\tilde \psi_k|\psi}}}\nonumber \\
     &=& \sqrt{1-\braket{\tilde\psi_k|\psi}} .
 \end{eqnarray}
 where in the second line, we have used the fact that $\braket{\psi_k|\psi}$ is real and positive and for the inequality in the third line, we have used the fact that $\braket{\psi_k|\psi} \leq 1$. Now,
 \begin{eqnarray}
     1-\braket{\tilde \psi_k|\psi} &=& \sum_{|\bm n| > k} |\psi_{\bm{n}}|^2 \nonumber \\
     &= & \sum_{|\bm{n}| < k} |\psi_{\bm{n}}|^2 s^{|\bm n|} s^{- |\bm n|} \nonumber \\
     &<& \frac{1}{s^k} \sum_{|\bm{n}|<k} |\psi_{\bm{n}}|^2 s^{|\bm{n}|} \nonumber \\
     &<& \frac{s^E}{s^k},
 \end{eqnarray}
 where for the inequality in the third line, we have used $s^{- |\bm n|} < 1/s^k$ for $s > 1$ and $|\bm n| > k$ and for the inequality in the fourth line, we have used $\sum_{|\bm{n}|>k} |\psi_{\bm{n}}|^2 s^{|\bm n|} \leq \sum_{\bm n} |\psi_{\bm{n}}|^2 s^{|\bm n|} < s^E$. Putting this in Eq.\ref{appeq:eff_descr_a} gives 
\begin{equation}
     \frac12 \|\ket{\psi}\bra{\psi} - \ket{\psi_k}\bra{\psi_k}\|_1 < \sqrt{\frac{s^E}{s^k}}.
\end{equation}
Therefore, if we want
\begin{eqnarray}
    \sqrt{\frac{s^E}{s^k}} &=& \epsilon, \nonumber\\
\frac{s^E}{s^k} &=& \epsilon^2, \nonumber \\
\frac{s^E}{\epsilon^2} &=& s^k.
\end{eqnarray}
Taking log on both sides, we get the required Fock state cutoff as
\begin{equation}
    k = \frac{\log(s^E)+ \log(1/\epsilon^2)}{\log(s)} = E+ \log_s(1/\epsilon^2).
\end{equation}
\end{proof}
\section{Bound on exponential-energy operator over Gaussian and energy-preserving non-Gaussian unitary gates}\label{appsec:exp_energy_bound energy simplified}
\noindent We first restate Theorem \ref{theo:exp_energy_bound_simplified}:
\begin{theo}\label{apptheo:exp_energy_bound_simplified}
  Given an $m$-mode quantum state $\ket{\psi} \in \mathcal C_L$ (Definition \ref{defi:G_nG}), with its associated Gaussian unitary gates $\hat G_i$ characterized by displacement vector $\bm{\alpha}_i$ and squeezing vector $\bm{r}_i$, $\forall i \in (1,\dots,L)$, then defining
    \begin{eqnarray}
        |\alpha|&:=& \max_{i,k} |\alpha_{ik}|,\hspace{5mm}
        r := \max_{i,k} r_{ik},
    \end{eqnarray} 
    $\forall i \in (1,\dots,L),k \in (1,\dots,m)$, if we choose $t_L > 1$ according to the scheme
    \begin{equation}\label{appeq:choice_t}
        t_i =1 +  \frac{t_{i-1} - 1}{e^{2r} + t_{i-1}}, \forall i \in (1,\dots,L),
    \end{equation}
    with $t_0 = 1 + 2/e^{2r}$, then it follows that
    \begin{eqnarray}\label{appeq:exp_energy_bound}
        \bra{\psi}t_L^{\hat N}\ket{\psi} &<& e^{mL^2(|\alpha|^2 + 28 r + 9)} \nonumber \\&<&  t_L^{mL^2 e^{2(2r+1)L} (|\alpha|^2 + 28 r +9)}.
    \end{eqnarray}
\end{theo}
\noindent For completeness, we also recall the definition of $\mathcal C_L$:
\begin{defi} \label{appdefi:G_nG}
    We define the set of $m$-mode states $\mathcal C_L$ as follows:
    \begin{equation}
        \mathcal C_L  = \{\ket{\psi}: \ket{\psi} = \hat G_L \hat{\chi}_L \hat G_{L-1} \hat{\chi}_{L-1}\dots\hat{G}_1 \ket{0}^{\otimes m}\},
    \end{equation}
    where $\hat \chi_1,\dots,\hat \chi_L$ are energy-preserving non-Gaussian unitary gates, and $\hat G_1,\dots,\hat G_L$ are arbitrary Gaussian unitary gates such that
    \begin{equation}
        \hat G_i = \hat V_i \hat D(\bm \alpha_i)\hat S(\bm{r}_i)\hat U_i
    \end{equation}
    where $\bm{\alpha}_i = (\alpha_{i1},\alpha_{i2},\dots,\alpha_{im}) \in \mathbb{C}^m, \bm{r}_i = (r_{i1},r_{i2},\dots,r_{im}) \in \mathbb{R}_{+}^m, \forall i \in \{1,\dots,L\}$, whereas $\hat U_i$ and $\hat V_i$ are energy-preserving Gaussian unitary gates.
\end{defi}
\noindent To prove Theorem \ref{theo:exp_energy_bound_simplified}, we first prove results on the growth of the exponetiated energy operator under the squeezing and displacement operator, and prove an analogous result to Theorem \ref{theo:exp_energy_bound_simplified} for a circuit composed of $L$ layers of single-mode squeezing and displacement operators.
\subsection{Exponential-energy operator under squeezing and displacement operators}\label{appsec:squeezing_increase}
\noindent To quantify the growth of the exponential-energy operator under the squeezing gate, we have the following Lemma:
\begin{lem}[Exponential-energy operator under multimode squeezing]\label{applem:squeeze_exponent} Given an $m$-mode quantum state $\ket{\psi}$ evolving under the $m$-mode squeezing operator $\hat S(\bm{r})$, where $\bm{r} = (r_1,r_2,\dots,r_m) \in \R_+^m \forall i \in (1,\dots,m)$, and the number operator $\hat N$, then it follows that
\begin{align}
    \bra{\psi}\hat S^\dagger(\bm{r}) t^{\hat{N}}\hat S(\bm{r})\ket{\psi}< g(r,s,t)^m \bra{\psi}s^{\hat N}\ket{\psi},
\end{align}
where
\begin{equation}\label{appeq:g}
    g(r,s,t):= s\times\left(\frac{e^{2r} + s_{r} e^{r}(s-1)}{e^{-2r} - s_{r} e^{-r}(s-1)}\right)^{1/2 }\times \frac{1}{(s-1) - (t-1)(e^{2r} + s_{r} e^{r}(s-1))},
\end{equation}
and $r= \max_i r_i ,\forall i \in (1,2,\dots,m)$, provided that
\begin{equation}\label{appeq:f}
    1 < t < f(r,s):= 1 + \frac{s-1}{e^{2r} + s_{r} e^{r} (s-1)} < s < \frac{1}{t_r}.
\end{equation}
Therefore, choosing
\begin{equation}
    t = 1 + \frac{s-1}{e^{2r} + 1},
\end{equation}
we get that
\begin{equation}
    \bra{\psi}\hat S^\dagger(\bm{r}) t^{\hat{N}}\hat S(\bm{r})\ket{\psi}< \left(\frac{s}{s-1}\times\frac{(e^{2r}+1)^2}{(e^{-r} - s_r(s-1))^2}\right)\bra{\psi}s^{\hat N}\ket{\psi}.
\end{equation}
\end{lem}
\begin{proof}
    Expanding $\ket{\psi}$ in the Fock basis
    \begin{equation}
        \ket{\psi} = \sum_{\bm{i} \geq \bm {0}} \psi_{\bm{i}} \ket{\bm{i}},
    \end{equation}
    where $\bm{i} = (i_1,i_2,\dots,i_m)$, then we can write
    \begin{equation}
        \hat S(\bm{r}) \ket{\psi} = \sum_{\bm{j} \geq \bm{0}} a_{\bm{j}} \ket{\bm{j}},
    \end{equation}
    where $a_{\bm{j}} = \sum_{\bm{i} \geq 0} \psi_{\bm{i}} \bra{\bm{j}}\hat S(\bm{r})\ket{\bm{i}}$. Therefore,
    \begin{eqnarray}
        \bra{\psi}\hat S^\dagger(\bm{r}) t^{\hat{N}}\hat S(\bm{r})\ket{\psi} &=& \sum_{\bm{j}\geq \bm{0}} t^{|\bm{j}|}|a_{\bm{j}}|^2 \nonumber \\
        &=& \sum_{\bm{j}\geq \bm{0}} t^{|\bm{j}|}\left| \sum_{\bm{i} \geq 0} \psi_{\bm{i}} \bra{\bm{j}}\hat S(\bm{r})\ket{\bm{i}}\right|^2 \nonumber \\
        &=& \sum_{\bm{j}\geq \bm{0}}t^{|\bm{j}|} \left| \sum_{\bm{i} \geq \bm{0}} \psi_{\bm{i}} s^{|\bm{i}|/2} s^{-|\bm{i}|/2} \bra{\bm{j}}\hat S(\bm{r})\ket{\bm{i}}\right|^2 \nonumber \\
        &\leq& \sum_{\bm{j} \geq \bm{0}} t^{|\bm{j}|} \sum_{\bm{i}\geq \bm{0}} \frac{|\bra{\bm{j}}\hat S(\bm{r})\ket{\bm{i}}|^2}{s^{|\bm{i}|}} \sum_{\bm{i} \geq \bm{0}} |\psi_{\bm{i}}|^2 s^{|\bm{i}|} \nonumber \\
        &=& \bra{\psi}s^{\hat N}\ket{\psi} \sum_{\bm{j} \geq \bm{0}} t^{|\bm{j}|} \sum_{\bm{i}\geq \bm{0}} \frac{|\bra{\bm{j}}\hat S(\bm{r})\ket{\bm{i}}|^2}{s^{|\bm{i}|}} \nonumber \\
        &=& \bra{\psi}s^{\hat N}\ket{\psi} \prod_{k=1}^m \sum_{j_k \geq 0} t^{j_k} \sum_{i_k \geq \bm{0}} \frac{|\bra{j_k}\hat S(r_k)\ket{i_k}|^2}{s^{i_k}},
    \end{eqnarray}
    where $|\bm {j}| = \sum_{k =1}^m j_k, |\bm {i}| = \sum_{k =1}^m i_k$ and the inequality in the fourth line follows from the Cauchy-Schwartz inequality and for the last line, we have used the fact that $|\bra{\bm{j}}\hat S(\bm{r})\ket{\bm{i}}|^2 = \prod_{k=1}^m |\bra{j_k}\hat S(r_k)\ket{i_k}|^2$. We then have the following result:
    
    \begin{lem}\label{applem:sum_squeezing}
     Given $t,s > 1$, with
\begin{equation}
   t < f(r,s) < s < \frac{1}{t_r},
   \end{equation}
(assuming $r \geq 0$), we have that
    \begin{eqnarray}
        \sum_{j \geq 0} t^j \sum_{i \geq 0} (1/s)^i |\bra{j}S(r)\ket{i}|^2 &\leq& g(r,s,t).
    \end{eqnarray}
    \end{lem}
\noindent The proof of Lemma \ref{applem:sum_squeezing} is given at the end of the section. We apply Lemma to $\sum_{j_k \geq 0} t^{j_k} \sum_{i_k \geq \bm{0}} \frac{|\bra{j_k}\hat S(r_k)\ket{i_k}|^2}{s^{i_k}},$ for all $ k \in (1,\dots, m)$ and note that both $1/t_r$ and $f(r,s)$ (for $s > 1$) are decreasing functions in r, the condition on $t$ and $s$ becomes
\begin{equation}\label{appeq:t_s_cond_squeezing}
   1 < t < f(r,s) < s < \frac{1}{t_{r}}.
\end{equation}
Finally we note that for $t$ and $s$ satisfying Eq.~\ref{appeq:t_s_cond_squeezing}, $g(r,s,t)$ is an increasing function of $r$, and we finally get
\begin{align}
    \bra{\psi}\hat S^\dagger(\bm{r}) t^{\hat{N}}\hat S(\bm{r})\ket{\psi}< g(r,s,t)^m \bra{\psi}s^{\hat N}\ket{\psi}.
\end{align}
Choosing
\begin{equation}
    t = 1 + \frac{s-1}{e^{2r}+1},
\end{equation}
we get that
\begin{eqnarray}
    \bra{\psi}\hat S^\dagger(\bm{r}) t^{\hat{N}}\hat S(\bm{r})\ket{\psi} &<& \left(\frac{s}{s-1} \left(\frac{e^{2r} + s_r e^r(s-1)}{e^{-2r} - s_r e^{-r}(s-1)}\right)^{1/2} \frac{e^{2r}+1}{e^{2r}(e^{-2r} - s_r e^{-r}(s-1))}\right)^m \bra{\psi}s^{\hat N}\ket{\psi} \nonumber \\
    &=& \left(\frac{s(e^{2r}+1)(e^{2r} + s_re^r(s-1))^{1/2}(e^{-2r} - s_r e^{-r}(s-1))^{1/2}}{e^{2r}(s-1)(e^{-2r} - s_r e^{-r}(s-1))^2} \right)^m \bra{\psi}s^{\hat N}\ket{\psi} \nonumber \\
    &< & \left(\frac{s(e^{2r}+1)^2}{e^{2r}(s-1)(e^{-2r} - s_r e^{-r}(s-1))^2} \right)^m \bra{\psi}s^{\hat N}\ket{\psi} \nonumber \\
    &=& \left(\frac{s}{s-1} \times \frac{(e^{2r}+1)^2}{(e^{-r} - s_r(s-1))}^2 \right)^m \bra{\psi}s^{\hat N}\ket{\psi},
\end{eqnarray}
where in the third line, we have used the inequality
\begin{equation}
    e^{-2r} - s_r e^{-r} (s-1) < e^{2r} + s_r e^{r} (s-1) < e^{2r} + 1
\end{equation}
for $s < 1/t_r$.
\end{proof}
\noindent Also, from \cite[Lemma 13]{cerf2025}, for the growth of the exponential-energy operator under the displacement gate, it follows that
\begin{lem}[Exponential-energy operator under multimode displacement operator]\label{applem:displace_exponent} Given an $m$-mode quantum state $\ket{\psi}$ evolving under the $m$-mode squeezing operator $\hat D(\bm{\alpha})$, where $\bm{\alpha} = (\alpha_1,\alpha_2,\dots,\alpha_m) \in \C^m$ and the number operator $\hat N$, then it follows that
    \begin{equation}
        \bra{\psi}\hat D^\dagger (\bm{\alpha}) t^{\hat N} \hat D(\bm{\alpha})\ket{\psi} < \frac{s^m}{(s-t)^m} \exp\left(m|\alpha|^2\frac{(s-1)(t-1)}{(s-t)}\right) \bra{\psi}s^{\hat n}\ket{\psi},
    \end{equation}
    where $|\alpha| = \max_i |\alpha|_i, \forall i \in (1,\dots,m)$provided that $1 <t < s$. Therefore, choosing 
    \begin{equation}
        t = 1 + \frac{s-1}{s},
    \end{equation}
    we get
    \begin{equation}
        \bra{\psi}\hat D^\dagger (\bm{\alpha}) t^{\hat N} \hat D(\bm{\alpha})\ket{\psi} < \left(\left(\frac{s}{s-1}\right)^2 e^{|\alpha|^2}\right)^m \bra{\psi}s^{\hat N}\ket{\psi}.
    \end{equation}
\end{lem}
\noindent Combining these two Lemmas allows us to bound the exponential-energy operator of a bosonic quantum state evolved under $L$ layers of Gaussian and arbitrary boson-number preserving non-Gaussian gates, as explained in the section below.

\bigskip

\noindent\textit{Proof of Lemma \ref{applem:sum_squeezing}}:
To prove Lemma \ref{applem:sum_squeezing}, we first prove another Lemma:
\begin{lem}\label{applem:Fock_sq_thermal_overlap}
Given a thermal state
\begin{equation}
    \rho_{\bar n} = (1+\bar n)^{-1} \sum_{ i \geq 0}\left(\frac{\bar n}{1+\bar n}\right)^i \ket{i}\bra{i},
\end{equation}
a squeezing operator $S(r)$ with $r > 0$ such that 
\begin{equation}
    \bar n e^{-2r} - e^{-r}s_r > 0 \rightarrow \bar n > \frac{s_r}{e^{-r}},
\end{equation}
and a Fock state $\ket{j}$, we have that
\begin{eqnarray}
    \bra{j}S(r)\rho_{\bar n} S^{\dagger}(r)\ket{j} &=& \Tr[\ket{j}\bra{j} S(r) \rho_{\bar n} S^\dagger(r)] \nonumber \\ &\leq& \left(\frac{\bar n e^{2r} + s_re^{r}}{\bar n e^{2r} - s_re^{-r}}\right)^{1/2} \frac{1}{1 + \bar n e^{2r} + s_re^{r}} \left(\frac{\bar n e^{2r} + s_re^{r}}{1 + \bar n e^{2r} + s_re^{r}}\right)^m. \nonumber \\
\end{eqnarray}    
\end{lem} 
\begin{proof}
    To prove this, note that the overlap between two states can be written in terms of the Glauber-Sudarshan P function of one and the Husimi Q function of the other as \cite{Kim1989}
    \begin{equation}
        \Tr[\ket{j}\bra{j} S(r) \rho_{\bar n} S^\dagger(r)] = \frac{1}{\pi}\int d\alpha P_{S(r) \rho_{\bar n} S^\dagger(r)} (\alpha) Q_{\ket{j}\bra{j}}(\alpha),
    \end{equation}
    Now, from \cite[Eq.~4.23]{Kim1989}, if $n e^{-2r} - e^{-r}s_r > 0$,
    \begin{equation}
        P_{S(r) \rho_{\bar n} S^\dagger(r)} (\alpha) = \frac{\exp\left(- \frac{\alpha_x^2}{\bar n e^{-2r} - s_re^{-r}} - \frac{\alpha_p^2}{\bar n e^{{2r}} + s_re^r}\right)}{\pi(\bar n e^{2r} + s_re^{r})^{1/2}(\bar n e^{-2r} - s_r e^{-r})^{1/2}}
    \end{equation}
Now, since
\begin{equation}
    \bar n e^{2r} + s_r e^{r} > \bar n e^{-2r} - s_r e^{-r},
\end{equation}
we have that $\exp(-\alpha_x^2/(\bar n e^{-2r} - s_r e^{-r})) \leq \exp(-\alpha_x^2/(\bar n e^{2r} + s_re^{r}))$ and hence
\begin{eqnarray}
    P_{S(r) \rho_{\bar n} S^\dagger(r)} (\alpha) &\leq& \left(\frac{\bar n e^{2r} + s_re^{r}}{\bar n e^{-2r} - s_r e^{-r}} \right)^{1/2} \frac{\exp\left(- \frac{\alpha_x^2}{\bar n e^{2r} + s_re^{r}} - \frac{\alpha_p^2}{\bar n e^{{2r}} + s_r e^r}\right)}{\pi(\bar n e^{2r} + s_re^{r})} \nonumber \\ &=& \left(\frac{\bar n e^{2r} + s_re^{r}}{\bar n e^{-2r} - s_r e^{-r}} \right)^{1/2}  P_{\rho_{\bar n e^{2r }+ s_r e^{r}}} (\alpha).
\end{eqnarray}
Furthermore, since $Q_{\ket{m}\bra{m}} (\alpha)=\frac1\pi|\langle\alpha|m\rangle|^2 \geq 0 \;\forall \alpha$, we have that
\begin{eqnarray}
    \Tr[\ket{j}\bra{j} S(r) \rho_{\bar n} S^\dagger(r)] &=& \frac{1}{\pi}\int d\alpha P_{S(r) \rho_{\bar n} S^\dagger(r)} (\alpha) Q_{\ket{j}\bra{j}}(\alpha) \nonumber \\
    &\leq& \left(\frac{\bar n e^{2r} + s_re^{r}}{\bar n e^{-2r} - s_r e^{-r}} \right)^{1/2} \frac{1}{\pi}\int d\alpha P_{\rho_{\bar n e^{2r} + s_re^r}} (\alpha) Q_{\ket{j}\bra{j}}(\alpha) \nonumber \\
    &=& \left(\frac{\bar n e^{2r} + s_r e^{r}}{\bar n e^{-2r} - s_r e^{-r}} \right)^{1/2} \Tr[\ket{j}\bra{j} \rho_{\bar n e^{2r} + s_r e^r}] \nonumber \\
    &=& \left(\frac{\bar n e^{2r} + s_r e^{r}}{\bar n e^{-2r} - s_r e^{-r}} \right)^{1/2}  \frac{1}{1 + \bar n e^{2r} + s_r e^r} \left( \frac{\bar n e^{2r} + s_r e^r}{1 + \bar n e^{2r} + s_r e^r}\right)^{j}. \nonumber \\
\end{eqnarray}
\end{proof}
\noindent With this Lemma, we are ready to prove Lemma \ref{applem:sum_squeezing}. The given expression is
    \begin{eqnarray}
        \sum_{i \geq 0} t^i \sum_{j \geq 0} (1/s)^j|\bra{j}S(r)\ket{i}|^2 &=& \sum_{i \geq 0} t^i \sum_{j \geq 0} (1/s)^j \bra{j}S(r)\ket{i}\bra{i}S^\dagger(r)\ket{j} \nonumber \\
        &=& \sum_{j \geq 0} t^j \bra{j} S(r)\left( \sum_{i \geq 0} (1/s)^i \ket{i}\bra{i}\right) S^\dagger(r) \ket{j}.
    \end{eqnarray}
    Note that since $s \geq 1$, up to normalization $\sum_{i \geq 0} (1/s)^i \ket{i}\bra{i}$ is a thermal state $\rho_{\bar n}$ such that $\bar n/(1 + \bar n) := 1/s$ or $\bar n = 1/(s-1)$. Therefore,
    \begin{eqnarray}
        \sum_{j \geq 0} t^j \sum_{i \geq 0} (1/s)^i |\bra{j}S(r)\ket{i}|^2 &=& (1 + \bar n) \sum_{j \geq 0} t^j \bra{j} S(r) \rho_{\bar n} S^\dagger(r) \ket{j},
    \end{eqnarray}
    so for  $\bar n e^{-r} - s_r > 0$ or equivalently $s < 1 + e^{-r}/s_r  = 1/t_r$, we use Lemma \ref{applem:Fock_sq_thermal_overlap} to write 
    \begin{eqnarray}
        \sum_{j \geq 0} t^j \sum_{i \geq 0} (1/s)^i |\bra{j}S(r)\ket{i}|^2 &\leq& (1 + \bar n) \left(\frac{\bar n e^{2r} + s_r e^{r}}{\bar n e^{-2r} - s_r e^{-r}} \right)^{1/2} \nonumber \\ && \times  \frac{1}{1 + \bar n e^{2r} + s_r e^r} \sum_{j \geq 0} t^j \left( \frac{\bar n e^{2r} + s_r e^r}{1 + \bar n e^{2r} + s_r e^r}\right)^{j}. \nonumber \\
    \end{eqnarray}
    So if
    \begin{equation}
        t < \frac{1 + \bar n e^{2r} + s_r e^r}{ \bar n e^{2r} + s_r e^r} = \frac{s-1 + e^{2r} + s_r e^{r}(s-1)}{ e^{2r} + s_r e^{r}(s-1)} = 1 + \frac{s-1}{e^{2r}+ s_r e^r(s-1)},
    \end{equation} we get
    \begin{eqnarray}
        \sum_{j \geq 0} t^j \sum_{i \geq 0} (1/s)^i |\bra{j}S(r)\ket{i}|^2 &\leq& (1 + \bar n) \left(\frac{\bar n e^{2r} + s_re^{r}}{\bar n e^{-2r} - s_r e^{-r}} \right)^{1/2} \nonumber \\
        &&\hspace{-15mm}\times \frac{1}{(1 + \bar n e^{2r} + s_r e^r)(1 - t (\bar n e^{2r} + s_re^{r})/(1+ \bar n e^{2r} + s_re^{r}))}. \nonumber \\
    \end{eqnarray}
    Finally, putting $\bar n = 1/(s-1)$, we get that
    \begin{eqnarray}
        \sum_{j \geq 0} t^j \sum_{i \geq 0} (1/s)^i |\bra{j}S(r)\ket{i}|^2 &\leq& s\times \left(\frac{e^{2r} + s_r e^r (s-1)}{e^{-2r} - s_r e^{-r}(s-1)}\right)^{1/2}\nonumber \\
        &\times& \frac{1}{(s-1) - (t-1)(e^{2r} + s_r e^r(s-1))}
    \end{eqnarray}
\subsection{Bound on the exponential-energy operator over multiple layers of displacement and squeezing gates}
To provide an upper bound on the expectation value of the exponential-energy operator for the quantum states prepared by a bosonic circuit composed of the $m$-mode vacuum fed into Gaussian and energy-preserving non-Gaussian gates, we first prove the result for a ``simpler'' output state and then explain in the next subsection why this also gives the answer for the general case:
\begin{theo}\label{apptheo:displaced_squeezed_exp_energy_simplified}
    Given an $m$-mode quantum state
    \begin{equation}
        \ket{\psi} = \underbrace{(\hat D(\bm{\alpha}) \hat S(\bm{r})) (\hat D(\bm{\alpha}) \hat S(\bm{r}))\dots (\hat D(\bm{\alpha}) \hat S(\bm{r}))}_{L \text{ times}}\ket{0}^{\otimes m},
    \end{equation}
    with $\bm{\alpha} = (\alpha_1,\dots,\alpha_m) \in \mathbb{C}^m$ and $\bm{r} = (r_1,r_2,\dots,r_m) \in \mathbb{R}^m$ with $r_i \geq 0, \forall i$. Then defining
    \begin{equation}
        r = \max_k r_k, \hspace{5mm} |\alpha| = \max_{k} |\alpha_k|, \forall k \in (1,\dots,m),
    \end{equation}
    if we choose $t_L > 1$ according to the scheme
    \begin{equation}\label{appeq:t_choice}
        t_i =1 +  \frac{t_{i-1} - 1}{e^{2r} + t_{i-1}}, \forall i \in (1,\dots,L),
    \end{equation}
    with $t_0 = 1 + 2/e^{2r}$, then it follows that
    \begin{equation}
        \bra{\psi}t_L^{\hat N}\ket{\psi} < e^{mL^2(|\alpha|^2 + 28 r + 9)}.
    \end{equation}
\end{theo}
\begin{proof}
    Denoting
    \begin{equation}
        \ket{\psi_i} := \hat D(\bm{\alpha}) \hat S(\bm{r}) \ket{\psi_{i-1}}, \forall i \in (1,\dots,L),
    \end{equation}
    with $\ket{\psi_0} = \ket{0}^{\otimes m}$, according to Lemmas \ref{applem:squeeze_exponent} and \ref{applem:displace_exponent}, choosing $t_{L} = 1 + \frac{\tilde t_L - 1}{\tilde t_L}$, with $\tilde t_L = 1 + \frac{t_{L-1} - 1}{e^{2r}+1}$, it follows that
    \begin{eqnarray}\label{appeq:imp}
        \bra{\psi_{L-1}}\hat S^\dagger (\bm{r}) \hat D^\dagger (\bm{\alpha}) t_{L}^{\hat N}\hat D(\bm{\alpha}) \hat S(\bm{r})\ket{\psi_{L-1}} &<& \left(\left(\frac{\tilde t_L}{\tilde t_L - 1}\right)^2 e^{|\alpha|^2}\right)^m \bra{\psi_{L-1}}\hat S^\dagger (\bm{r}) \tilde{t}_L^{\hat N}\hat S(\bm{r})\ket{\psi_{L-1}} \nonumber \\
        &<& \left(\left(\frac{\tilde t_L}{\tilde t_L - 1}\right)^2 e^{|\alpha|^2} \frac{t_{L-1}}{t_{L-1} - 1} \times \frac{(e^{2r} + 1)^2}{(e^{-r} - s_r(t_{L-1}-1))^2}\right)^m \nonumber \\
        && \hspace{5mm}\times\bra{\psi_{L-1}} t_{L-1}^{\hat N}\ket{\psi_{L-1}}, \nonumber \\
        &=& \left(\left(\frac{1 + \frac{t_{L-1} - 1}{e^{2r}+ 1}}{\frac{t_{L-1} - 1}{e^{2r}+ 1}}\right)^2 e^{|\alpha|^2} \frac{t_{L-1}}{t_{L-1} - 1} \times \frac{(e^{2r} + 1)^2}{(e^{-r} - s_r(t_{L-1}-1))^2}\right)^m \nonumber \\
        && \hspace{5mm}\times\bra{\psi_{L-1}} t_{L-1}^{\hat N}\ket{\psi_{L-1}}, \nonumber \\
        &<& \left(\left(\frac{e^{2r} + t_{L-1}}{t_{L-1} - 1}\right)^3 e^{|\alpha|^2} \times \frac{(e^{2r} + 1)^2}{(e^{-r} - s_r(t_{L-1}-1))^2}\right)^m \nonumber \\
        && \hspace{5mm}\times\bra{\psi_{L-1}} t_{L-1}^{\hat N}\ket{\psi_{L-1}},
    \end{eqnarray}
    where for the last inequality, we have used the fact that
    \begin{equation}
        t_{L-1} < e^{2r} + t_{L-1}.
    \end{equation}
    Now, $\left(\frac{e^{2r} + t_{L-1}}{t_{L-1} - 1}\right)^3$ is a decreasing function of $t_{L-1}$, whereas $\frac{1}{(e^{-r} - s_r(t_{L-1}-1))^2}$ is an increasing function of $t_{L-1}$, therefore to upper bound $\left(\frac{e^{2r} + t_{L-1}}{t_{L-1} - 1}\right)^3 e^{|\alpha|^2} \times \frac{(e^{2r} + 1)^2}{(e^{-r} - s_r(t_{L-1}-1))^2}$, we need both an upper bound and lower bound on $t_{L-1}$. Now, from Eq.~\ref{appeq:t_choice} (note that combining $t_{L} = 1 + \frac{\tilde t_L - 1}{\tilde t_L}$, with $\tilde t_L = 1 + \frac{t_{L-1} - 1}{e^{2r}+1}$ gives Eq.~\ref{appeq:t_choice}), its clear that $t_{L-1} < t_0$, and we further note that
    \begin{equation}
        t_{L-1} = 1 + \frac{t_{L-2} - 1}{e^{2r} + t_{L-2}} > 1 + \frac{t_{L-2} - 1}{e^{2r} + t_{0}} > 1 + \frac{t_0 - 1}{(e^{2r}+t_0)^{L-1} } > 1 + \frac{t_0 - 1}{(e^{2r}+t_0)^{L}}.
    \end{equation}
    Putting this in Equation \ref{appeq:imp} gives
    \begin{eqnarray}
    \bra{\psi_{L-1}}\hat S^\dagger (\bm{r}) \hat D^\dagger (\bm{\alpha}) t_{L}^{\hat N}\hat D(\bm{\alpha}) \hat S(\bm{r})\ket{\psi_{L-1}}    &<& \left(\left(\frac{e^{2r} + 1 + \frac{t_0 -1 }{(e^{2r}+t_0)^L}}{\frac{t_0 - 1}{(e^{2r} + t_0)^L}}\right)^3 e^{|\alpha|^2} \times \frac{(e^{2r} + 1)^2}{(e^{-r} - s_r(t_0-1))^2}\right)^m \nonumber\\
    &&\hspace{5mm} \times \bra{\psi_{L-1}}t_{L-1}^{\hat N}\ket{\psi_{L-1}}\nonumber\\
    &<& \left(\left(\frac{(e^{2r} + 1)(e^{2r}+t_0)^L}{t_0 - 1}+1\right)^3 e^{|\alpha|^2}\times \frac{(e^{2r} + 1)^2}{(e^{-r} - s_r(t_0-1))^2}\right)^m\nonumber\\
    &&\hspace{5mm} \times \bra{\psi_{L-1}}t_{L-1}^{\hat N}\ket{\psi_{L-1}}.
    \end{eqnarray}
    Choosing $t_0 = 1 + 2/e^{2r}$ then gives
\begin{eqnarray}
    \bra{\psi_{L-1}}\hat S^\dagger (\bm{r}) \hat D^\dagger (\bm{\alpha}) t_{L}^{\hat N}\hat D(\bm{\alpha}) \hat S(\bm{r})\ket{\psi_{L-1}} &<& \left(\left(\frac{e^{2r}(e^{2r}+1)(e^{2r} + 1  + 2/e^{2r})^L} {2} + 1\right)^3 e^{|\alpha|^2 + 6r} (e^{2r} + 1)^2\right)^m \nonumber\\
    &&\hspace{5mm} \times \bra{\psi_{L-1}}t_{L-1}^{\hat N}\ket{\psi_{L-1}}.
\end{eqnarray}
Finally, nothing that $1 \leq e^{2r} \leq e^{2r}\frac{(e^{2r}+1)}{2} (e^{2r}+ t_0)^L$ and that $e^{2r} + 1  + 2/e^{2r} \leq e^{2r} + 3 \leq 3 e^{2r} + 3 \leq 6e^{2r}$, we get that
\begin{eqnarray}\label{appeq:to_chain}
    \bra{\psi_{L-1}}\hat S^\dagger (\bm{r}) \hat D^\dagger (\bm{\alpha}) t_{L}^{\hat N}\hat D(\bm{\alpha}) \hat S(\bm{r})\ket{\psi_{L-1}} &<& \left((e^{2r})^3(2e^{2r})^3(6e^{2r})^{3L}e^{|\alpha|^2 + 6r}(2e^{2r})^2\right)^m  \bra{\psi_{L-1}}t_{L-1}^{\hat N}\ket{\psi_{L-1}} \nonumber \\
    &=& \left(32\times 6^{3L}e^{|\alpha|^2 + 22r + 6Lr}\right)^m \bra{\psi_{L-1}}t_{L-1}^{\hat N}\ket{\psi_{L-1}} \nonumber \\
    &<& \left(32 \times 216 \times e^{|\alpha|^2 + 28r}\right)^{mL}\bra{\psi_{L-1}}t_{L-1}^{\hat N}\ket{\psi_{L-1}} \nonumber \\
    &<& e^{mL(|\alpha|^2 + 28r + 9)} \bra{\psi_{L-1}}t_{L-1}^{\hat N}\ket{\psi_{L-1}},
\end{eqnarray}
where in the third line, we have used the assumption that $L \geq 1$ and in the last line, we are using the fact that $e^{\log(32\times216)} < e^9$. Chaining Eq.~\ref{appeq:to_chain} over the $L$ layers, we get
\begin{eqnarray}
    \bra{\psi_{L-1}}\hat S^\dagger (\bm{r}) \hat D^\dagger (\bm{\alpha}) t_{L}^{\hat N}\hat D(\bm{\alpha}) \hat S(\bm{r})\ket{\psi_{L-1}} &<& e^{mL(L+1)/2(|\alpha|^2 + 28r + 9)} \bra{\psi_{0}}t_{0}^{\hat N}\ket{\psi_{0}} \nonumber \\ &\leq& e^{mL^2(|\alpha|^2 + 28r + 9)} \bra{\psi_{0}}t_{0}^{\hat N}\ket{\psi_{0}} = e^{mL^2(|\alpha|^2 + 28r + 9)},
\end{eqnarray}
where for the last equality, we have used the fact that for $\ket{\psi_0} = \ket{0}^{\otimes m}$, $\bra{\psi_{0}}t_{0}^{\hat N}\ket{\psi_{0}} = 1$.
\end{proof}
\noindent Theorem \ref{apptheo:displaced_squeezed_exp_energy_simplified} also allows us to find a bound on the exponential-energy operator for circuit composed of Gaussian and energy-preserving non-Gaussian unitary gates, as detailed below.
\subsection{Proof of Theorem \ref{theo:exp_energy_bound_simplified}}
\begin{proof}
        Denoting
    \begin{equation}
        \ket{\psi_i} = \hat G_i \hat{\chi}_i \ket{\psi_{i-1}}, \forall i \in (1,\dots,L)
    \end{equation}
    with $\ket{\psi_0}$ being the $m$-mode vacuum state, we see using Lemmas \ref{applem:squeeze_exponent} and \ref{applem:displace_exponent} that for a constant $t_i > 1$,
    \begin{eqnarray}\label{appeq:remove_passive}
        \bra{\psi_i} t_i^{\hat N} \ket{\psi_i} &=& \bra{\psi_{i-1}}\chi_{i}^\dagger \hat U_i^\dagger \hat S^{\dagger}(\bm{r}_i) \hat D^\dagger (\bm \alpha_i) \hat V_i^\dagger t_i^{\hat N} \hat V_i D(\bm \alpha_i) \hat S(\bm r_i) \hat U_i \hat \chi_i\ket{\psi_{i-1}} \nonumber \\
        &=& \bra{\psi_{i-1}}\chi_i^\dagger \hat U_i^\dagger \hat S^{\dagger}(\bm{r}_i) \hat D^\dagger (\bm \alpha_i) t_i^{\hat N} D(\bm \alpha_i) \hat S(\bm r_i) \hat U_i \hat \chi_i \ket{\psi_{i-1}} \nonumber \\
        &<& \frac{\tilde{t}_i^m}{(\tilde t_i - t_i)^m} \exp\left(m|\alpha_i|^2 \frac{(\tilde t_i - 1)(t_i -1)}{\tilde t_i - t_i}\right) g(r_{i},t_{i-1},\tilde t_i)^m \bra{\psi_{i-1}}\hat \chi_{i}^\dagger \hat U_i^\dagger  t_{i-1}^{\hat N} \hat U_i \hat \chi_i \ket{\psi_{i-1}} \nonumber \\
        &=& \frac{\tilde{t}_i^m}{(\tilde t_i - t_i)^m} \exp\left(m|\alpha_i|^2 \frac{(\tilde t_i - 1)(t_i -1)}{\tilde t_i - t_i}\right) g(r_{i},t_{i-1},\tilde t_i)^m \bra{\psi_{i-1}} t_{i-1}^{\hat N} \ket{\psi_{i-1}},
    \end{eqnarray}
    provided that
    \begin{equation}
        1 < t_i < \tilde {t_i} < f(r_{i},t_{i-1}) < t_{i-1} < \frac{1}{t_{r_{i}}},
    \end{equation}
    where
    \begin{equation}
        |\alpha_i| = \max_{k}|\alpha_{ik}|, \qquad r_{i} = \max_k r_{ik},\qquad \forall k \in (1,\dots,L),
    \end{equation}
    and the functions $g$ and $f$ are given by Eqs.~\ref{appeq:g} and Eqs.~\ref{appeq:f} respectively. We see from Eq.~\ref{appeq:remove_passive} that the energy-preserving Gaussian unitary gates associated to each of the Gaussian unitary $\hat G_i$, and the energy-preserving non-Gaussian unitary gates $\hat{\chi}_i$ do not change the expectation value of the exponential-energy operator, since they commute with it. Therefore, if we instead consider the state
    \begin{equation}
        \ket{\psi_{\mathrm{red}}} = \left(\prod_{i=1}^L \hat D(\bm \alpha_i) \hat S(\bm{r}_i)\right)\ket{0}^{\otimes m},
    \end{equation}
    the upper bound on $\bra{\psi_{\mathrm{red}}}t_L^{\hat  N}\ket{\psi_{\mathrm{red}}}$ also gives the upper bound on $\bra{\psi}t_L^{\hat  N}\ket{\psi}$. Now,
    \begin{equation}
        \frac{\tilde{t}_i^m}{(\tilde t_i - t_i)^m} \exp\left(m|\alpha_i|^2 \frac{(\tilde t_i - 1)(t_i -1)}{\tilde t_i - t_i}\right) g(r_{i},t_{i-1},\tilde t_i)^m
    \end{equation}
     is an increasing function of $|\alpha_i|$ and $r_{i}$, whereas $f(r_{i},t_{i-1})$ and $1/t_{r_{i}}$ are decreasing functions of $r_{i}$, therefore, defining
     \begin{equation}
         |\alpha| = \max_{i,k} |\alpha_{ik}|, \hspace{5mm} r = \max_{i,k} r_{ik},
     \end{equation}
     $\forall i \in (1,\dots,L), k\in (1,\dots,m)$. Therefore, choosing $t_0,\tilde t_1, t_1,\dots, t_{L-1}, \tilde t_L,t_L$, such that
     \begin{equation}\label{appeq:choice_t_new}
         1  < t_i < \tilde t_i < f(r,t_{i-1}) < t_{i-1} < \frac{1}{t_{r}},
     \end{equation}
     $\forall i \in (1,\dots,L)$, it follows from Eq.~\ref{appeq:remove_passive} that
     \begin{eqnarray}
         \bra{\psi_{i,\mathrm{red}}} t_{i}^{\hat N} \ket{\psi_{i,\mathrm{red}}} < \frac{\tilde{t}_i^m}{(\tilde t_i - t_i)^m} \exp\left(m|\alpha|^2 \frac{(\tilde t_i - 1)(t_i -1)}{\tilde t_i - t_i}\right) g(r,t_{i-1},\tilde t_i)^m&& \nonumber\\ && \hspace{-35mm}\bra{\psi_{i-1}} t_{i-1}^{\hat N} \ket{\psi_{i-1}},
     \end{eqnarray}
     with 
     \begin{equation}
         \ket{\psi_{i,\mathrm{red}}} = \hat{D}(\bm{\alpha}_i)\hat S(\bm{r}_i) \ket{\psi_{i-1,\mathrm{red}}},
     \end{equation}
     $\forall i \in (1,\dots,L)$, with $\ket{\psi_{0,\mathrm{red}}}$ being the $m$-mode vacuum state. Therefore, denoting 
     \begin{equation}
        \ket{\phi} = \underbrace{(\hat D(\bm{\alpha}_{\max}) \hat S(\bm{r}_{\max})) (\hat D(\bm{\alpha}_{\max}) \hat S(\bm{r}_{\max})) \dots (\hat D(\bm{\alpha}_{\max}) \hat S(\bm{r}_{\max}))}_{L \text{ times}}\ket{0}^{\otimes m},
    \end{equation}
    where $\bm r_{\max} = (r,r,\dots,r) \in \R_+^m$ and $\bm \alpha_{\max} = (\alpha,\alpha,\dots,\alpha) \in \mathbb{C}^m$, where $\alpha$ is the complex number corresponding to $|\alpha| $, the upper bound on $\bra{\phi}t_L^{\hat N}\ket{\phi}$ with the choice of $t_L$ governed by Eq.~\ref{appeq:choice_t_new} gives an upper bound on $\bra{\psi_{\mathrm{red}}} t_L^{\hat N}\ket{\psi_{\mathrm{red}}}$ and hence equivalently on $\bra{\psi} t_L^{\hat N}\ket{\psi}$. To find an upper bound on $\bra{\phi}t_L^{\hat N}\ket{\phi}$ with the choice of $t_L$ governed by Eq.~\ref{appeq:choice_t_new}, we can directly use Theorem \ref{apptheo:displaced_squeezed_exp_energy_simplified}. This gives the first part of the statement of Theorem \ref{theo:exp_energy_bound_simplified}. We then write
    \begin{equation}
        \bra{\psi}t_L^{\hat N}\ket{\psi} < e^{mL^2 (|\alpha|^2 + 28 r +9)} = t_{L}^{mL^2 (|\alpha|^2 + 28 r +9) \log_{t_L} e} = t_{L}^{\frac{mL^2 (|\alpha|^2 + 28 r +9)}{\log(t_L)}}.
    \end{equation}
    Given that $1/\log(1+t_L)$ is a decreasing function of $t_L$, with $t_L > 1 +  \frac{t_0 - 1}{(e^{2r} +  t_0)^L}$, where $\frac{t_0 - 1}{(e^{2r} +  t_0)^L} \leq 1$, and using the fact that $1/\log(1+x) \leq 2/x$ for $x\leq 1$, we get that
    \begin{equation}
        \bra{\psi}t_L^{\hat N}\ket{\psi} < t_L^{\frac{2mL^2(|\alpha|^2 + 28 r +9)}{t_0 - 1} (e^{2r} + t_0)^L}.
    \end{equation}
    Finally, putting $t_0 = 1 + 2/e^{2r}$, and using the fact that $(e^{2r} + t_0)^L = (e^{2r} + 1 + 2/e^{2r})^L \leq (3e^{2r})^L \leq e^{2rL + 2L}$ and that $e^{2r} \leq e^{2rL}$ for $L\geq 1$, we get that
    \begin{equation}
         \bra{\psi}t_L^{\hat N}\ket{\psi} < t_L^{mL^2 e^{2(2r+1)L} (|\alpha|^2 + 28 r +9)}.
    \end{equation}
\end{proof}
\section{Comparing Fock space cutoffs based on energy operator vs exponential-energy operator bounds}\label{appsec:cutoff_comparisons}
This section focuses on two results: we compare the effective Fock space dimension of states produced by $L$ layers of Gaussian and energy-preserving non-Gaussian gates (Definition \ref{defi:G_nG}) that we get through the application of Theorem \ref{theo:exp_energy_bound_simplified} and Lemma \ref{lem:efficient_description} against the effective Fock space dimension we get from the Gentle Measurement Lemma \cite{Winter1999}. In addition, we demonstrate why the exponent of $t_L$ can be interpreted as the energy of the state. To do this, we need a bound on energy of the state prepared by $L$ layers of Gaussian and energy-preserving non-Gaussian gates. This is given in the following Theorem (Note that this is the multimode generalization of \cite[Lemma 3.2]{chabaud2025energy}):
\begin{theo}\label{theo:bound_energy}
   Given an $m$-mode state $\ket{\psi} \in \mathcal C_L$ (Definition \ref{defi:G_nG}), with its associated Gaussian unitary gates $\hat G_i$ characterized by displacement vector $\bm{\alpha}_i$ and squeezing vector $\bm{r}_i$, $\forall i \in (1,\dots,L)$, then defining
    \begin{eqnarray}
        |\alpha|&:=& \max_{i,k} |\alpha_{ik}|,\hspace{3mm}
        r := \max_{i,k} r_{ik}, \nonumber \\
    \end{eqnarray} 
    $\forall i \in (1,\dots,L),k \in (1,\dots,m)$, then it follows that
    \begin{equation}\label{eq:energy_operator_bound}
        \bra{\psi} \hat N \ket{\psi} \leq m(A^L - 1) \left(1  + \frac{|\alpha|^2 }{A - 1}\right),
    \end{equation}
    where 
    \begin{equation}
        A := e^{2r} + 4|\alpha|e^{r}.
    \end{equation}
\end{theo}
\noindent The proof of Theorem \ref{theo:bound_energy} is given at the end of the section. 

Therefore, combining Theorem \ref{theo:bound_energy} with the standard energy-based truncation (Lemma \ref{lem:energy_cutoff}), we see that the best known result previously would imply a Fock state cutoff of $\mathcal{O}(m \exp(L))/\epsilon^2$ for an $\epsilon$-approximate description of $\ket{\psi} \in \mathcal C_L$, whereas by combining Theorem \ref{theo:exp_energy_bound_simplified} and Lemma \ref{lem:efficient_description}, and noting that $1/\log(t_L) < \mathcal{O}(e^{Lr})$ (since from Eq.~\ref{eq:choice_t}, $t_L > 1 +  (t_0-1)/(e^{2r} + t_0)^L$ with $t_0 = 1 + 2/e^{2r}$ and the fact that $1/\log(1+x) \leq 2/x$ for $x< 1$), our results show that a Fock state cutoff of $\mathcal O(\exp(L)(mL^2 +\log(4/\epsilon^2)))$ is sufficient for an $\epsilon$-approximate description of $\ket{\psi}$. Therefore, we have reduced the dependence on the desired precision from $\mathcal{O}(1/\epsilon^2)$ to $\mathcal{O}(\log(1/\epsilon^2))$, an exponential improvement. 

At this point, we note that since by Jensen's inequality $t_L^{\bra{\psi}\hat N\ket{\psi}} \leq \bra{\psi} t_L^{\hat N}\ket{\psi}$, therefore $ \bra{\psi} t_L^{\hat N}\ket{\psi} \leq t_L^{E_L}$ implies $\bra{\psi} \hat N\ket{\psi} \leq E_L$. To understand the other direction on how the bound on $t_L^{\hat N}$ over the states $\ket\psi \in \mathcal C_L$ can be thought of as an exponential energy bound, we first take the case of a single-mode coherent state $\ket{\alpha}$ as an illustration, where we see that $\bra{\alpha}s^{\hat N}\ket{\alpha} = \exp((s-1)|\alpha|^2) = s^{\frac{s-1}{\log s}|\alpha|^2}$ and $E = \bra{\alpha}\hat N\ket{\alpha} = |\alpha|^2$ . Therefore, we can write $\bra{\alpha}s^{\hat N}\ket{\alpha} = \exp((s-1)E)$ and we see that, up to a scaling factor in the exponent, the bound on $s^{\hat N}$ can be thought of as an exponential energy bound.

For the general case of states $\ket{\psi} \in \mathcal C_L$, note that the exponent of $t_L$ on the right hand side of Eq.~\ref{eq:exp_energy_bound} and the upper bound on the number operator in Eq.~\ref{eq:energy_operator_bound} follow similar behavior in terms of the scaling in $L$ and $m$ (both scale as $\mathcal O(\exp(L)\poly(m)))$. Furthermore, focusing on the most significant contribution in the exponent of $t_L$ in Eq.~\ref{eq:exp_energy_bound}, $me^{2(r+1)L}$ and that of the energy operator $m(e^{2r}+ 4|\alpha|r)^L$, we see that the boundedness of either the term in the exponent of $t_L$ or the energy would impose similar conditions on $|\alpha|$ and $r$. This suggests that the term in the exponent of $t_L$ effectively captures the state's energy growth.

From the example of the single-mode coherent states and looking at the scaling of the term in the exponent of $t_L$ and the upper bound on the energy operator, we interpret that the term in the exponent of $t_L$ in Eq.~\ref{eq:exp_energy_bound} can be thought of as the notion of energy of the state. 

\noindent We now give the proof of Theorem \ref{theo:bound_energy}:
\begin{proof}
    Denoting
    \begin{equation}
        \ket{\psi_i} = \hat V_i \hat D(\bm{\alpha}_i)\hat S(\bm{r}_i) \hat U_i \hat \chi_i\ket{\psi_{i-1}}, \forall i \in (1,\dots,L)
    \end{equation}
    $\forall i \in (1,\dots,L)$, with $\ket{\psi_0} = \ket{0}^{\otimes m}$, we see that
    \begin{eqnarray}
        \bra{\psi_{i}} \hat N + m \mathbb{I}\ket{\psi_i} &=&  \bra{\psi_{i-1}} \hat \chi_i^\dagger \hat U_i^\dagger \hat S^\dagger(\bm{r}_i) \hat D^\dagger(\bm{\alpha}_i) \hat V_i^\dagger (\hat N + m \mathbb{I}) \hat V_i \hat D(\bm{\alpha}_i) \hat S(\bm{r}_i) \hat U_i \hat \chi_i \ket{\psi_{i-1}} \nonumber \\
        &= &\bra{\psi_{i-1}} \hat \chi_i^\dagger \hat U_i^\dagger \hat S^\dagger(\bm{r}_i) \hat D^\dagger(\bm{\alpha}_i) (\hat N + m \mathbb{I}) D(\bm{\alpha}_i) \hat S(\bm{r}_i) \hat U_i \hat \chi_i \ket{\psi_{i-1}},
    \end{eqnarray}
    where we have used the fact that that energy-preserving Gaussian unitary $\hat V_i$ commutes with the number operator. Writing $\hat N = \sum_{k=1}^m (\hat q_k^2 + \hat p_k^2-1)/2$, denoting $\alpha_{ik} = (\alpha_{ik})_q + i (\alpha_{ik})_p, \forall i \in (1,\dots,L);k \in(1,\dots,m)$, and evolving the quadratures under multimode-squeezing and displacement operator, we get that
    \begin{eqnarray}\label{appeq:energy_growth}
        \bra{\psi_i} \hat N + m \mathbb{I}\ket{\psi_i} &=& \bra{\psi_{i-1}}\hat \chi_i^\dagger \hat U_i^\dagger \left(\sum_{k=1}^m \frac{(e^{r_k} \hat q_k + \sqrt{2}(\alpha_{ik})_q)^2 + (e^{-r_k} \hat p_k + \sqrt{2}(\alpha_{ik})_p)^2}{2} + \frac{m}{2} \mathbb{I}\right) \hat U_i \hat \chi_i \ket{\psi_{i-1}} \nonumber \\
        &=& \bra{\psi_{i-1}}\hat \chi_i^\dagger \hat U_i^\dagger \left(\sum_{k=1}^m \frac{e^{2r_k}\hat q_k^2 + e^{-2r_k}\hat p_k^2} {2} + \frac{m}{2} \mathbb{I}\right) \hat U_i \hat \chi_i \ket{\psi_{i-1}} + \sum_{k=1}^m |\alpha_{ik}|^2 \nonumber \\
        && +  \bra{\psi_{i-1}}\hat \chi_i^\dagger \hat U_i^\dagger \left(\sum_{k=1}^m \sqrt{2 } (\alpha_{ik})_q e^{r_k} \hat q_k + \sqrt{2} (\alpha_{ik})_p e^{-r_k} \hat p_k \right) \hat U_i \hat \chi_i \ket{\psi_{i-1}} \nonumber \\
        &\overset{(\romannumeral 1)}{\leq}& e^{2r} \bra{\psi_{i-1}}\hat \chi_i^\dagger \hat U_i^\dagger (\hat N + m \mathbb{I})\hat U_i \hat \chi_i \ket{\psi_{i-1}} + m |\alpha|^2 \nonumber \\
        && \hspace{-8mm}+  \sum_{k=1}^m \sqrt{2} e^{r_k} |(\alpha_{ik})_q| |\bra{\psi_{i-1}}\hat \chi_i^\dagger \hat U_i^\dagger \hat q_k \hat U_i \hat \chi_i \ket{\psi_{i-1}}| + \sqrt{2} e^{-r_k} |(\alpha_{ik})_p| |\bra{\psi_{i-1}}\hat \chi_i^\dagger \hat U_i^\dagger \hat p_k \hat U_i \hat \chi_i \ket{\psi_{i-1}}| \nonumber \\
        &\overset{(\romannumeral 2)}{\leq}& e^{2r} \bra{\psi_{i-1}}\hat \chi_i^\dagger \hat U_i^\dagger (\hat N + m \mathbb{I})\hat U_i \hat \chi_i \ket{\psi_{i-1}} + m |\alpha|^2 \nonumber \\
        &&+\sum_{k=1}^m 4e^{r_k} |\alpha_{ik}| \sqrt{\bra{\psi_{i-1}}\hat \chi_i^\dagger \hat U_i^\dagger \left(\frac{\hat q_k^2 +  \hat p_k^2 + \mathbb{I}}{2}\right) \hat U_i \hat \chi_i \ket{\psi_{i-1}}} \nonumber \\
        &\overset{(\romannumeral 3)}{\leq}& e^{2r} \bra{\psi_{i-1}}\hat \chi_i^\dagger \hat U_i^\dagger (\hat N + m \mathbb{I})\hat U_i \hat \chi_i \ket{\psi_{i-1}} + m |\alpha|^2 + 4e^{r} |\alpha|\sum_{k=1}^m \bra{\psi_{i-1}} \hat \chi_i^\dagger \hat U_i^\dagger (\hat N_k + 
 \mathbb{I})\hat U_i \hat \chi_i \ket{\psi_{i-1}}\nonumber \\
        &=& (e^{2r} + 4e^r |\alpha|) \bra{\psi_{i-1}}\hat \chi_i^\dagger \hat U_i^\dagger (\hat N + m \mathbb{I})\hat U_i \hat \chi_i \ket{\psi_{i-1}} + m |\alpha|^2 \nonumber \\
        &&\hspace{-10mm} \overset{(\romannumeral 4)}{=} (e^{2r} + 4e^r |\alpha|) \bra{\psi_{i-1}} \hat N + m \mathbb{I} \ket{\psi_{i-1}} + m |\alpha|^2 := A \bra{\psi_{i-1}} \hat N + m \mathbb{I}\ket{\psi_{i-1}} + B,
    \end{eqnarray}
    where in $(\romannumeral 1)$, we are using
    \begin{eqnarray}
        \bra{\psi_{i-1}}\hat \chi_i^\dagger \hat U_i^\dagger \left(\sum_{k=1}^m \frac{e^{2r_k}\hat q_k^2 + e^{-2r_k}\hat p_k^2}{2} + \frac{m}{2}\mathbb{I}\right) \hat U_i \hat \chi_i \ket{\psi_{i-1}} &\leq& e^{2r} \bra{\psi_{i-1}}\hat \chi_i^\dagger \hat U_i^\dagger \left(\sum_{k=1}^m \frac{\hat q_k^2 + \hat p_k^2}{2} + \frac{m}{2} \mathbb{I}\right)\hat U_i \hat \chi_i \ket{\psi_{i-1}}, \nonumber \\
        e^{r_k} (\alpha_{ik})_q \bra{\psi_{i-1}}\hat \chi_i^\dagger \hat U_i^\dagger \hat q_k \hat U_i \hat \chi_i \ket{\psi_{i-1}}&\leq& e^{r_k} |(\alpha_{ik})_q| |\bra{\psi_{i-1}}\hat \chi_i^\dagger \hat U_i^\dagger \hat q_k \hat U_i \hat \chi_i \ket{\psi_{i-1}}|, \nonumber \\
        e^{-r_k} (\alpha_{ik})_p \bra{\psi_{i-1}}\hat \chi_i^\dagger \hat U_i^\dagger \hat p_k \hat U_i \hat \chi_i \ket{\psi_{i-1}}&\leq&  e^{-r_k} |(\alpha_{ik})_p| |\bra{\psi_{i-1}}\hat \chi_i^\dagger \hat U_i^\dagger \hat p_k \hat U_i \hat \chi_i \ket{\psi_{i-1}}|.
    \end{eqnarray}
    In $(\romannumeral 2)$, we use the Cauchy-Schwartz inequality, $r_k \geq 0$ and $|(\alpha_{ik})_q|,|(\alpha_{ik})_p| \leq |\alpha_{ik}|$ to write
    \begin{eqnarray}
        \sqrt{2}e^{r_k} |(\alpha_{ik})_q| |\bra{\psi_{i-1}}\hat \chi_i^\dagger \hat U_i^\dagger \hat q_k \hat U_i \hat \chi_i \ket{\psi_{i-1}}| &\leq& 2e^{r_k} |\alpha_{ik}| \sqrt{\bra{\psi_{i-1}}\hat \chi_i^\dagger \hat U_i^\dagger\left( \frac{\hat q_k^2 +  \hat p_k^2 + \mathbb{I}}{2}\right)\hat U_i \hat \chi_i \ket{\psi_{i-1}}}, \nonumber \\
        \sqrt{2}e^{-r_k} |(\alpha_{ik})_p| |\bra{\psi_{i-1}}\hat \chi_i^\dagger \hat U_i^\dagger \hat p_k \hat U_i \hat \chi_i \ket{\psi_{i-1}}| &\leq& 2e^{r_k} |\alpha_{ik}| \sqrt{\bra{\psi_{i-1}}\hat \chi_i^\dagger \hat U_i^\dagger\left( \frac{\hat q_k^2 +  \hat p_k^2 + \mathbb{I}}{2}\right)\hat U_i \hat \chi_i \ket{\psi_{i-1}}}. \nonumber \\
    \end{eqnarray}
    In $(\romannumeral 3)$, we use $e^{r_k} \leq e^r$, $|\alpha_{ik}| \leq |\alpha|$, $(\hat q_k^2 + \hat p_k^2 + \mathbb{I})/2 = \hat N_k + \mathbb{I}$ and that $\sqrt{\bra{\psi_{i-1}}\hat \chi_i^\dagger \hat U_i^\dagger (\hat N_k + \mathbb{I})\hat U_i \hat \chi_i \ket{\psi_{i-1}}} \leq  \bra{\psi_{i-1}}\hat \chi_i^\dagger \hat U_i^\dagger (\hat N_k + \mathbb{I})\hat U_i \hat \chi_i \ket{\psi_{i-1}}$, since $\bra{\psi_{i-1}}\hat \chi_i^\dagger \hat U_i^\dagger (\hat N_k + \mathbb{I})\hat U_i \hat \chi_i \ket{\psi_{i-1}} \geq 1$. Finally, in $(\romannumeral 4)$, we are using the fact that the energy-preserving unitary gates $\hat U_i$ and $\hat \chi_i$ commute with the number operator, and we denote $A:= (e^{2r} + 4e^{r}|\alpha|)$ and $B:= m |\alpha|^2$. Using Eq.~\ref{appeq:energy_growth}, we write
    \begin{eqnarray}
        \bra{\psi_L} \hat N + m \mathbb{I}\ket{\psi_L} &\leq& A \bra{\psi_{L-1}} (\hat N + m \mathbb{I})\ket{\psi_{L-1}} + B \nonumber \\
        &\leq& A^2 \bra{\psi_{L-2}} (\hat N + m \mathbb{I})\ket{\psi_{L-2}} + B(1+ A) \nonumber \\
        && \vdots \nonumber \\
        &\leq& A^L \bra{\psi_0} (\hat N + m \mathbb{I})\ket{\psi_0} + B\sum_{k=1}^{L-1}A^k \nonumber \\
        &=& m A^L + B\frac{A^L - 1}{A-1},
    \end{eqnarray}
    where in the last line, we have used the fact that $\ket{\psi_{0}}$ is the $m$-mode vacuum state. Therefore, since $\ket{\psi_L} = \ket{\psi}$, we can write
    \begin{equation}
        \bra{\psi} \hat N \ket{\psi} \leq (A^L - 1)\left(m + \frac{B}{A-1}\right) = m(A^L - 1)\left(1 + \frac{|\alpha|^2}{A-1}\right).
    \end{equation}
\end{proof}
\section{Learning bosonic quantum states with exponential-energy bound}\label{appsec:precision_efficient_learning}
\noindent In this section, we prove how the exponential improvement in effective dimension of the states $\psi \in \mathcal S_{s,mE}$ (Eq.~\ref{eq:set_bounded}) via Lemma \ref{lem:efficient_description} leads to a substantial improvement in the sample complexity of learning these states, as formalized in the following Theorem:
\begin{theo}\label{apptheo:precision_efficient_learning}
  Given a state $\ket{\psi}$ belonging to the set $\mathcal S_{s,mE}$ given by Eq.~\ref{eq:set_bounded}, a number of samples 
    \begin{equation}
        \mathcal{O}\left(\frac{1}{\epsilon^2}\left(E + \frac{2\log_s(2/\epsilon)}{m}\right)^m\right),
    \end{equation}
    where $\mathcal O$ represents scaling up to log factors, is sufficient to construct the classical description of an estimator $\ket{\tilde \psi}$ such that 
    \begin{equation}
        \frac12 \|\ket{\psi}\bra{\psi} - \ket{\tilde \psi}\bra{\tilde \psi}\|_1 \leq \epsilon.
    \end{equation}
\end{theo}
\begin{proof}
   We note that the proof follows the proof steps of \cite[Theorem S35]{mele2024learning}, with an improved effective dimension. Given $2N' + 24 \log (2/\delta)$ copies of $\ket{\psi}$, where we will fix $N'$ later, we execute the POVM $\{\Pi_k,1-\Pi_k\}$, where $\Pi_k$ is the projector onto the Fock space with total boson number $\leq k$. If we choose 
   \begin{equation}
       k = mE + 2\log_s(2/\epsilon),
   \end{equation}
   then for a state $\ket{\psi} \in \mathcal S_{s,mE}$, from Lemma \ref{lem:efficient_description}, with probability
   \begin{eqnarray}
       \Tr[\Pi_k \ket{\psi}\bra{\psi}] = p_{succ} = 1- (1 - p_{succ}) \geq  1 - \frac{\epsilon^2}{4} \geq 1 - \frac14 = \frac34,
   \end{eqnarray}
   we get the post measurement state
   \begin{equation}
       \ket{\psi_k} = \frac{\Pi_k\ket{\psi}}{\Tr[\Pi_k \ket{\psi}\bra{\psi}]},
   \end{equation}
   and it follows that
   \begin{equation}
       \frac12\|\ket{\psi}\bra{\psi} - \ket{\psi_k}\bra{\psi_k}\|_1 \leq \frac{\epsilon}{2}.
   \end{equation}
   And with $p_{succ} \geq 3/4$, from \cite[Lemma S5]{mele2024learning}, we need $2N' + 24 \log(2/\delta)$ copies of $\ket{\psi}$ to obtain $N'$ copies of $\psi_k$ with the success probability $1 - \delta/2$. Since the dimension of $\ket{\psi_k}$ is given by
   \begin{equation}
     d_{eff} =  \binom{m + k}{m} \leq e^m \left(\frac{m + k}{m}\right)^m = \left(1 + E + \frac{2\log_s(2/\epsilon)}{m}\right)^m,
   \end{equation}
    where we have used $\binom{a}{b} \leq (ea/b)^b$ for $a \geq b > 0$, then from \cite{wright2016learn,mele2024learning}, using 
   \begin{equation}
       N' \geq  \frac{2^{20}}{\epsilon^2} \left(1 + E + \frac{2\log_s(2/\epsilon)}{m}\right)^m \log\left(\frac{4}{\delta}\right),
   \end{equation}
   copies of $\ket{\psi_k}$, we can build the classical description of an estimator $\ket{\tilde{\psi}}$ such that
   \begin{equation}
       \frac12\|\ket{\psi_k}\ket{\psi_k} - \ket{\tilde\psi}\bra{\tilde{\psi}}\|_1 \leq \frac{\epsilon}{2}
   \end{equation}
   with success probability $ \geq 1 - \delta/2$
   Putting everything together and using the triangle inequality, we see that using 
   \begin{equation}
       \frac{2^{21}}{\epsilon^2} \left(1 + E + \frac{2\log_s(2/\epsilon)}{m}\right)^m \log\left(\frac{4}{\delta}\right) + 24 \log\left(\frac{2}{\delta}\right)
   \end{equation}
   copies of $\ket{\psi}$, we can we can build the classical description of an estimator $\ket{\tilde{\psi}}$ such that
   \begin{equation}
      \frac12 \|\ket{\psi}\ket{\psi} - \ket{\tilde\psi}\bra{\tilde{\psi}}\|_1 \leq \epsilon
   \end{equation}
   with success probability $\geq 1 - \delta$. We keep the most significant scalings in the statement of the Theorem.
\end{proof}
\section{Classical simulation by tracking Fock state amplitudes}\label{appsec:Fock_space_simulation}
\noindent The exponential improvement in effective descriptions of states prepared by Gaussian and energy-preserving non-Gaussian gates allows for efficient classical simulation of a certain class of these circuits by simply tracking the Fock space amplitudes. We illustrate one such simulation problem here, given by the formal version of Theorem \ref{theo:Fock_state_simulation} of the main text, as well as illustrate why efficient classical simulation using this method would not have been possible if we were the using the naive energy based truncations (Lemma \ref{lem:energy_cutoff}):
\begin{theo}
  Let $n$ be a size parameter. Given an $m$-mode vacuum state $\ket{0}^{\otimes m}$ evolving through the unitary gate
    \begin{equation}
        \hat U = \hat G_L \hat \chi_L \dots \hat G_1 \hat \chi_1,
    \end{equation}
    where $\hat G_1, \dots, \hat G_L$ are arbitrary Gaussian gates, and $\hat \chi_1, \dots, \hat \chi_L$ are energy-preserving non-Gaussian gates such that the analytical formula for the action of $\chi_1,\dots,\chi_L$ on Fock states can be computed in time at most exponential in number of bosons in the Fock state. Then, given that $m = \mathcal O(\log(n)), L = \mathcal O(\poly(n))$ and there exists a constant $s > 1$ such that $s^{\hat N}$ is bounded by $\mathcal{O}(\poly(n))$ throughout the computation,  then the probabilities generated by Gaussian state projections on the first $k$ mode of the output state
    \begin{equation}
        P = \left|\bra{G_1\dots G_k}\otimes \mathbb{I}_{m-k} \hat U \ket{0}^{\otimes m}\right|^2
    \end{equation}
    can be estimated up to $1/\mathcal{O}(\poly(n))$ additive precision in time $\mathcal O(\poly(n))$, which allows for approximate strong simulation and in particular, approximate sampling.
\end{theo}
\begin{proof}
   The algorithm works by truncating the state produced by the action of the Gaussian gate on the input finite superposition of Fock states up to a finite boson number cutoff with truncation error quantified by Lemma \ref{lem:efficient_description}. We compute exactly the amplitudes of the truncated state using ideas from \cite{yao2024} and \cite{bjorklund2019fasthafnian}. 

    We first quantify the trace distance error introduced by doing a truncation whenever we encounter the Gaussian gate. We denote
    \begin{equation}
        \ket{\psi_i} = \hat G_i \hat \chi_i \ket{\psi_{i-1}}, \forall i \in (1,\dots,L)
    \end{equation}
    with $\ket{\psi_0} = \ket{0}^{\otimes m}$. Suppose that we have approximated $ \ket{\psi_{i-1}}$ with a normalized finite superposition of Fock states $\ket{\phi_{i-1}}$ (Note that $\ket{\phi_0}$ is the $m$-mode vacuum state). We want to estimate the trace distance between $\ket{\psi_i}$ and normalized truncation of $\hat G_i \hat{\chi}_i \ket{\phi_{i-1}}$ up to total boson number $k$.
    \begin{align}\label{appeq:Gauss_truncation_error}
        &\frac12\left\|\ket{\psi_i}\bra{\psi_i} - \frac{\Pi_k \hat G_i \hat{\chi}_i \ket{\phi_{i-1}} \bra{\phi_{i-1}}\hat \chi_i^\dagger \hat G_i^\dagger \Pi_k}{\Tr[\Pi_k \hat G_i \hat{\chi}_i \ket{\phi_{i-1}} \bra{\phi_{i-1}}\hat \chi_i^\dagger \hat G_i^\dagger]}\right\|_1 \nonumber \\ &\hspace{10mm}\overset{(\romannumeral 1)}{\leq} \frac12\left\|\ket{\psi_i}\bra{\psi_i} - \frac{\Pi_k \ket{\psi_{i}}\bra{\psi_{i}} \Pi_k}{\Tr[\Pi_k \ket{\psi_{i}}\bra{\psi_{i}}]}\right\|_1 + \frac12 \left\|\frac{\Pi_k \ket{\psi_{i}}\bra{\psi_{i}} \Pi_k}{\Tr[\Pi_k \ket{\psi_{i}}\bra{\psi_{i}}]} - \frac{\Pi_k \hat G_i \hat{\chi}_i \ket{\phi_{i-1}} \bra{\phi_{i-1}}\hat \chi_i^\dagger \hat G_i^\dagger \Pi_k}{\Tr[\Pi_k \hat G_i \hat{\chi}_i \ket{\phi_{i-1}} \bra{\phi_{i-1}}\hat \chi_i^\dagger \hat G_i^\dagger]} \right\|_1 \nonumber \\
        &\hspace{10mm}\overset{(\romannumeral 2)}{\leq}\sqrt{\frac{S}{s^k}} + \frac12 \left\|\frac{\Pi_k \ket{\psi_{i}}\bra{\psi_{i}} \Pi_k}{\Tr[\Pi_k \ket{\psi_{i}}\bra{\psi_{i}}]} - \hat G_i \hat{\chi}_i \ket{\phi_{i-1}} \bra{\phi_{i-1}}\hat \chi_i^\dagger \hat G_i^\dagger \right\|_1 \nonumber \\
        &\hspace{10mm}\overset{(\romannumeral 3)}{\leq}\sqrt{\frac{S}{s^k}} + \frac12\left\| \frac{\Pi_k \ket{\psi_{i}}\bra{\psi_{i}} \Pi_k}{\Tr[\Pi_k \ket{\psi_{i}}\bra{\psi_{i}}]} - \ket{\psi_i}\bra{\psi_i}\right\|_1 + \frac12\left\|\ket{\psi_i}\bra{\psi_i} - \hat G_i \hat{\chi}_i \ket{\phi_{i-1}} \bra{\phi_{i-1}}\hat \chi_i^\dagger \hat G_i^\dagger\right\|_1 \nonumber \\
        & \hspace{10mm} \overset{(\romannumeral 4)}{\leq} 2 \sqrt{\frac{S}{s^k}} + \frac12\left\|\ket{\psi_{i-1}}\bra{\psi_{i-1}} -  \ket{\phi_{i-1}} \bra{\phi_{i-1}}\right\|_1,
    \end{align}
    where (\romannumeral 1) follows from the triangle inequality, (\romannumeral 2) follows from $\bra{\psi_i}s^{\hat N}\ket{\psi_i} < S$ combined with  Lemma \ref{lem:efficient_description} for the first term and the fact that
    \begin{eqnarray}
         \hspace{-15mm}\frac12\left\|\frac{\Pi_k \ket{\psi_{i}}\bra{\psi_{i}} \Pi_k}{\Tr[\Pi_k \ket{\psi_{i}}\bra{\psi_{i}}]} - \frac{\Pi_k \hat G_i \hat{\chi}_i \ket{\phi_{i-1}} \bra{\phi_{i-1}}\hat \chi_i^\dagger \hat G_i^\dagger \Pi_k}{\Tr[\Pi_k \hat G_i \hat{\chi}_i \ket{\phi_{i-1}} \bra{\phi_{i-1}}\hat \chi_i^\dagger \hat G_i^\dagger]} \right\|_1 &=& \sqrt{1 - \frac{|\braket{\psi_i^k|\phi_{i-1}^{G_i,\chi_i,k}}|^2}{\Tr[\Pi_k \ket{\psi_{i}}\bra{\psi_{i}}] \Tr[\Pi_k \hat G_i \hat{\chi}_i \ket{\phi_{i-1}} \bra{\phi_{i-1}}\hat \chi_i^\dagger \hat G_i^\dagger]}}\nonumber \\ &\leq& \sqrt{1 - \frac{|\braket{\psi_i^k|\phi_{i-1}^{G_i,\chi_i,k}}|^2}{\Tr[\Pi_k \ket{\psi_{i}}\bra{\psi_{i}}] }} \nonumber \\
         &=& \frac12 \left\|\frac{\Pi_k \ket{\psi_{i}}\bra{\psi_{i}} \Pi_k}{\Tr[\Pi_k \ket{\psi_{i}}\bra{\psi_{i}}]} - \hat G_i \hat{\chi}_i \ket{\phi_{i-1}} \bra{\phi_{i-1}}\hat \chi_i^\dagger \hat G_i^\dagger \right\|_1 \nonumber \\
    \end{eqnarray}
   for the second term, where we have denoted $\ket{\psi_{i}^k} := \Pi_k\ket{\psi_i}$ and $\ket{\phi_{i-1}^{G_i,\chi_i,k}} := \Pi_k \hat G_i \hat \chi_i \ket{\phi_{i-1}}$ and used the fact that $\Tr[\Pi_k \hat G_i \hat{\chi}_i \ket{\phi_{i-1}} \bra{\phi_{i-1}}\hat \chi_i^\dagger \hat G_i^\dagger] \leq 1$. Inequality (\romannumeral 3) again follows from the triangle inequality and finally, inequality (\romannumeral 4) follows from $\bra{\psi_i}s^{\hat N}\ket{\psi_i} < S$, Lemma \ref{lem:efficient_description} and the fact that unitaries do not change the trace distance. 

    From Eq.~\ref{appeq:Gauss_truncation_error}, if we make a Fock space truncation with $k = \log_s(S/\epsilon^2)$ each time we encounter a Gaussian gate, the output state is $2\epsilon L$ close in trace distance to the actual state. We set  $\epsilon = \epsilon'/L$ with $\epsilon' = 1/\mathcal O(\poly(n))$ so that the output state of the circuit is $1/O(\poly(n))$ close in trace distance to the actual state.

    Taking the description of $\ket{\phi_{i-1}}$ as the normalized superposition of $\mathcal O(\poly(n))$ Fock states with maximum boson number $k_{i-1} = \mathcal O(\log(\poly(n)))$ (the rest of the proof will clarify why this is the case), then given that the action of energy-preserving non-Gaussian gates on the Fock states can be described in time at most $e^{k_{i-1}} = \mathcal O(\poly(n))$ and since passive non-Gaussian gates do not change the boson number of a given Fock state and $\ket{\phi_{i-1}}$ is a finite superposition of Fock states with total boson number $\mathcal O(\log(\poly(n)))$ and number of terms $\mathcal O(\poly(n))$, the description of
    \begin{equation}
        \hat \chi_i \ket{\phi_{i-1}} = \sum_{|\bm q|\leq Q} c_{\bm m} \ket{\bm q}
    \end{equation}
    with $Q = \mathcal O(\log(\poly(n)))$ can be given in time $\mathcal O(\poly(n))$.
    
    To describe the action of the Gaussian unitary $\hat G_i$, from the discussion following Eq.~\ref{appeq:Gauss_truncation_error}, we first truncate $\hat G_i \hat \chi_i \ket{\phi_{i-1}}$ up to total boson number 
    \begin{equation}
        k = \log_s(S L^2/\epsilon'^2) = \mathcal O(\log(\poly(n)))
    \end{equation}
    to ensure that the approximated output state is $\epsilon' = 1/\mathcal O(\poly(n))$ close in trace distance to the actual state. Since $m,k = \mathcal O(\log(\poly(n)))$, the effective dimension (and equivalently, the number of terms) of $\Pi_k \hat G_i \hat \chi_i \ket{\phi_{i-1}}$ is 
    \begin{equation}
        d = \binom{m + k}{k} \leq \binom{2m}{m} \leq 4^m = \mathcal O(\poly(n)).
    \end{equation}
    Writing
    \begin{equation}
        \Pi_k\hat G_i \hat \chi_i \ket{\phi_{i-1}} = \sum_{|\bm p| \leq k} a_{\bm p} \ket{\bm p},
    \end{equation}
    with
    \begin{equation}
        a_{\bm p} = \sum_{|\bm q| \leq M} c_{\bm q} \bra{\bm p} \hat G_i \ket{\bm q}.
    \end{equation}
    Now, from \cite[Eq.~98]{yao2024}, $ \bra{\bm p} \hat G_i \ket{\bm q}$ is, up to a constant depending on the Gaussian unitary, the loop Hafnian of $|\bm p| + |\bm q|$ matrix whose entries depends on the unitary $\hat G_i$, divided by $p_1!p_2! \dots p_m! q_1! \dots q_m!$. In our case, $|\bm p|,|\bm q| = \mathcal O(\log(\poly(n)))$ and therefore $|\bm p| + |\bm q| = \mathcal O(\log(\poly(n)))$. From \cite{Borwein1985}, $i!$ can be computed in time $\mathcal O(i)$ (up to $\log$ factors) and therefore $p_1!p_2! \dots p_m! q_1! \dots q_m!$ can be computed in time $\mathcal O(m |\bm p|) = \mathcal O(\log(\poly(n)))$. Furthermore, from \cite{bjorklund2019fasthafnian}, the loop Hafnian of a matrix of size $n$ can be calculated in time $\mathcal O(n^3 2^{n/2})$, therefore $\bra{\bm p} \hat G_i \ket{\bm q}$ can be calculated in time $\mathcal O(\poly(n))$. Since $ a_{\bm p} = \sum_{|\bm q| \leq Q} c_{\bm q} \bra{\bm p} \hat G_i \ket{\bm q}$ with the sum having $d = \mathcal O(\poly(n))$ terms, therefore the time complexity to calculate $a_{\bm p}$ is given by $\mathcal O(\poly(n))$. For getting the description of $\Pi_k \hat G_i \hat \chi_i \ket{\phi_{i-1}}$, we need to calculate $d = \mathcal O(\poly(n))$ such $a_{\bm p}$'s, hence we can get the description of $\Pi_k \hat G_i \hat \chi_i \ket{\phi_{i-1}}$ in time $\mathcal{O}(\poly(n))$. And once we have this description, the normalization constant $\mathcal N = \sum_{|\bm p|\leq k} |a_{\bm p}|^2$ can be computed in time $\mathcal O(\poly(n))$ (since we have $\mathcal O(\poly(n))$ terms). We write $\ket{\phi_i} = \Pi_k \hat G_i \hat \chi_i \ket{\phi_{i-1}} /\sqrt{\mathcal N}$.

    Repeating this over the $L$ layers, we can write the description of the approximated output state $\ket{\tilde \psi}$ as the Fock state superposition with effective dimension $d = \mathcal O(\poly(n))$ and maximum boson number $k = \mathcal O(\log(\poly(n)))$ in time $\mathcal O(L\poly(n)) = \mathcal O(\poly(n))$ such that  $\ket{\tilde \psi}$ is $1/\mathcal{O}(\poly(n))$ close in trace distance to the actual output state $\ket \psi$. Finally, from \cite[Theorem 2]{chabaud2021corestates}, we can compute 
    \begin{equation}
        \tilde P = \left|\bra{G_1\dots G_k}\otimes \mathbb{I}_{m-k} \ket{\tilde \psi}\right|^2
    \end{equation}
    in time $\mathcal O(d^2 k^3 2^k + \poly(m)) = \mathcal O(\poly(n))$ (note that \cite[Theorem 2]{chabaud2021corestates} is for heterodyne measurements, but any single-mode Gaussian measurement can be written as a single-mode Gaussian gate followed by heterodyne measurement). And by the property of the trace distance, $\tilde P$ approximates 
    \begin{equation}
        P = \left|\bra{G_1\dots G_k}\otimes \mathbb{I}_{m-k} \ket{\psi}\right|^2
    \end{equation}
    up to additive precision $1/\mathcal{O}(\poly(n))$.
\end{proof}
\noindent To understand why the same method would not work efficient using the naive cutoff based on energy, we first note that by Jensen's inequality $s^{\bra{\psi}\hat N \ket{\psi}} \leq \bra{\psi}s^{\hat N}\ket{\psi}$ for all states $\ket{\psi}$, therefore a bound of $\mathcal O(\poly(n))$ on the exponential-energy throughout the computation would imply a bound of $\mathcal O(\log(\poly(n)))$ on the energy throughout the computation. Using the cutoff based on Energy in Eq.~\ref{appeq:Gauss_truncation_error}, we get that
\begin{equation}
    \frac12\left\|\ket{\psi_i}\bra{\psi_i} - \frac{\Pi_k \hat G_i \hat{\chi}_i \ket{\phi_{i-1}} \bra{\phi_{i-1}}\hat \chi_i^\dagger \hat G_i^\dagger \Pi_k}{\Tr[\Pi_k \hat G_i \hat{\chi}_i \ket{\phi_{i-1}} \bra{\phi_{i-1}}\hat \chi_i^\dagger \hat G_i^\dagger ]}\right\|_1 \leq 2\sqrt{\frac{E}{k}} + \|\ket{\psi_{i-1}}\bra{\psi_{i-1}} - \ket{\phi_{i-1}}\bra{\phi_{i-1}}\|_1.
\end{equation}
Therefore, to ensure that the trace distance after each layer increases by only $\epsilon'/L$, with $L = \mathcal{O}(\poly(n))$ and $\epsilon' = \mathcal O(1/\poly(n))$, we choose the boson number cutoff
\begin{equation}
    k = EL^2/\epsilon'^2 = \mathcal O(\poly(n)).
\end{equation}
that gives the effective dimension
\begin{equation}
    \binom{m+k}{m} = \mathcal{O}(\exp^{m + m\log(1 + k /m)}) = \mathcal O(\poly(n)^{\log(1 + \poly(n))}).
\end{equation}
The super-polynomial effective dimension of the state throughout the computation for the required inverse polynomial precision means that we cannot simple track the description of the state throughout the computation even approximately because we would have to calculate a superpolynomial number of terms after each layer, and this makes the classical simulation using the naive energy cutoff inefficient. 
\section{Decomposition of Kerr gates}\label{appsec:kerr_decomp}
In this section, we give the following Lemma about the decomposition of self-Kerr and cross-Kerr gates, first proved in \cite[Eqs. 8-12]{Jun-wei1996} for the self-Kerr gates and here we generalize it to cross-Kerr gates:
\begin{lem}\label{applem:kerr_decomp}
    Given integers $p\in \Z, q \in \Z_+^*$, if $q$ is even, then the self-Kerr gate $\hat \kappa(- p/q)$ can be expressed as
        \begin{equation}
    \hat \kappa\left(-\frac{p}{q}\right) = \sum_{j=0}^{q-1} g_{j,p,q}^{(e)} \exp\left(-i \frac{2\pi j}{q} \hat N_1\right),
\end{equation}
where
\begin{equation}\label{appeq:coeff_even}
    g_{j,p,q}^{(e)} = \frac{1}{q} \sum_{n=0}^{q-1} \exp\left(\frac{2\pi i j}{q} n\right) \exp\left(-\frac{i \pi p}{q}n^2\right).
\end{equation}
On the other hand, if $q$ is odd, then the self-Kerr gate $\hat \kappa(-p/q)$ can be expressed as
\begin{equation}
    \hat\kappa\left(-\frac{p}{q}\right) = \sum_{j=0}^{q-1} g_{j,p,q}^{(o)} \exp\left(-i \frac{\pi}{q} (2j+p) \hat N_1\right),
\end{equation}
where
\begin{equation}\label{appeq:coeff_odd}
    g_{j,p,q}^{(o)} = \frac{1}{q} \sum_{n=0}^{q-1} \exp\left(\frac{2\pi i j}{q} n\right) \exp\left(-\frac{i \pi p}{q}n(n-1)\right).
\end{equation}
Furthermore, for all integers $p\in \Z, q \in \Z_+^*$ the cross-Kerr gate can be decomposed as 
\begin{equation}
    \hat{c\kappa} \left(\frac{p}{q}\right) = \sum_{j,k = 0}^{q-1} g_{j,k,p,q} \exp\left(i \frac{2\pi}{q}(j \hat N_1 + k \hat N_2)\right),
\end{equation}
where 
\begin{equation}
    g_{j,k,p,q} = \frac{1}{q^2} \sum_{m,n=0}^{q-1} \exp\left(i\frac{2\pi p}{q} mn\right) \exp\left(-i \frac{2\pi}{q}(jm+kn)\right).
\end{equation}
\end{lem}
\begin{proof}
First, we note that
\begin{eqnarray}
    \exp\left(-i\pi \frac{p}{q} (\hat N_1+q)^2\right) &=&  (-1)^{pq} \exp\left(-i\pi \frac{p}{q} \hat N_1^2\right), \nonumber \\
    \exp\left(-i\pi \frac{p}{q} (\hat N_1+ q) (\hat N_1 + q - 1)\right) &=& (-1)^{p(q-1)} \exp\left(-i\pi \frac{p}{q} \hat N_1 (\hat N_1 - 1)\right).
\end{eqnarray}
Therefore, $\exp\left(-i\pi \frac{p}{q} \hat N_1^2\right) \left(\exp\left(-i\pi \frac{p}{q} \hat N_1 (\hat N_1 - 1)\right)\right)$ is periodic with a period of $q$ for even (odd) $q$. We use this periodicity of these self-Kerr gates to express them as a sum of phase-shifters periodic with respect to $q$:
\begin{eqnarray}
    \exp\left(-i\pi \frac{p}{q} \hat N_1^2\right) &=& \sum_{j=0}^{q-1} g_{j,p,q}^{(e)} \exp\left(-i \frac{2\pi j}{q} \hat N_1\right), \nonumber \\
    \exp\left(-i\pi \frac{p}{q} \hat N_1 (\hat N_1 - 1)\right) &=& \sum_{j=0}^{q-1} g_{j,p,q}^{(o)} \exp\left(-i \frac{2\pi j}{q} \hat N_1\right).
\end{eqnarray}
To find the coefficients $g_j^{(e)}$ and $g_j^{(o)}$, we do the inverse Fourier transform of the above equation and use
\begin{equation}\label{eq:kron_delta}
    \sum_{j=0}^{q-1} \exp\left(-\frac{2\pi i j}{q} N\right) = q \delta_{N,0}
\end{equation}
to get
\begin{eqnarray}
    g_{j,p,q}^{(e)} &=& \frac{1}{q} \sum_{n=0}^{q-1} \exp\left(\frac{2\pi i j}{q} n\right) \exp\left(-\frac{i \pi p}{q}n^2\right), \nonumber\\
    g_{j,p,q}^{(o)} &=& \frac{1}{q} \sum_{n=0}^{q-1} \exp\left(\frac{2\pi i j}{q} n\right) \exp\left(-\frac{i \pi p}{q}n(n-1)\right).
\end{eqnarray}
In the case of $pq$ being odd, we shift $\exp(+i\pi \hat N/p)$ on the right hand to obtain the expression for the cross-Kerr gates.

\noindent For the cross-Kerr gates, we note that
    \begin{eqnarray}
        \exp\left(i2\pi \frac{p}{q} (\hat N_1+q) \otimes \hat N_2\right) &=& \exp\left(i2\pi \frac{p}{q} \hat N_1 \otimes \hat N_2\right) \exp\left(i 2\pi p \hat N_2\right) = \exp\left(i2\pi \frac{p}{q} \hat N_1 \otimes \hat N_2\right), \nonumber \\
        \exp\left(i2\pi \frac{p}{q} \hat N_1 \otimes (\hat N_2+q)\right) &=& \exp\left(i2\pi \frac{p}{q} \hat N_1 \otimes \hat N_2\right).
    \end{eqnarray}
    Therefore, $\exp\left(i2\pi \frac{p}{q} \hat N_1 \otimes \hat N_2\right)$ is periodic with a period of $q$ with respect to both $\hat N_1$ and $\hat N_2$, therefore we hypothesize the following decomposition for it:
    \begin{equation}
        \exp\left(i2\pi \frac{p}{q} \hat N_1 \otimes \hat N_2\right) = \sum_{j,k=0}^{q-1} g_{j,k,p,q} \exp\left(i \frac{2\pi j}{q} \hat N_1 \right) \exp\left(i \frac{2\pi k}{q} \hat N_2 \right),
    \end{equation}
    and to find $g_{j,k,p,q}$, we again use the property
    \begin{equation}
        \sum_{j = 0}^{q-1} \exp \left( i \frac{2\pi j}{q} N\right) = q \delta_{N,0}
    \end{equation}
    and do the inverse Fourier transform of the equation to find
    \begin{equation}
        g_{j,k,p,q} = \frac{1}{q^2} \sum_{m,n=0}^{q-1} \exp\left(- i \frac{2\pi j}{q} m\right) \exp\left(-i\frac{2\pi k}{q}n\right) \exp\left(i \frac{2\pi p}{q} mn\right).
    \end{equation}
\end{proof}
\section{Trace distance error induced by ignoring self-Kerr gate with a small parameter}\label{appsec:eps_kerr_gate_td}
We restate Lemma \ref{lem:eps_kerr_gate_td}:
\begin{lem}
    Given a multimode state $\ket{\psi}$ and multimode number operator $\hat N$, such that  $\bra{\psi} e^{-i\epsilon \hat N_1^2} s^{\hat N} e^{i\epsilon \hat N_1^2} \ket{\psi} = \bra{\psi}  s^{\hat N}\ket{\psi} \leq s^E$, the trace distance between $\ket{\psi}$ and $e^{i\epsilon \hat N_1^2}\ket{\psi}$ is upper bounded by 
    \begin{equation}
        \frac12\|e^{i\epsilon \hat N_1^2} \ket{\psi}\bra{\psi}e^{-i\epsilon \hat N_1^2} - \ket{\psi}\bra{\psi}\|_1 = \mathcal O\left(\epsilon E^2 \log_s^2(1/\epsilon^2)\right)
    \end{equation}
\end{lem}
\begin{proof}
We write
\begin{eqnarray}
    \ket{\psi} &=& \sum_{\bm{n}} a_{\bm{n}} \ket{\bm{n}} \nonumber \\
    \ket{\phi} &:=& e^{i\epsilon \hat N_1^2} \ket{\psi} =  \sum_{\bm{n}} a_{\bm{n}} e^{i\epsilon n_1^2} \ket{\bm{n}}
\end{eqnarray}
    Given that 
    \begin{equation}
        \bra{\psi}  s^{\hat N}  \ket{\psi} = \bra{\phi}  s^{\hat N}  \ket{\phi}\leq s^E
    \end{equation}
    If we take the normalised projection of $\ket{\psi}$ on the Fock subspace with total bosons less than equal to $k$, that is, defining
    \begin{eqnarray}
        \ket{\psi_{\mathrm{red}}} &=& \frac{\sum_{|\bm{n}| \leq k} a_{\bm{n}} \ket{\bm{n}}}{\sqrt{\sum_{|\bm{j}| \leq k} |a_{\bm{j}}|^2}}, \nonumber \\
    \ket{\phi_{\mathrm{red}}} &=&  \frac{\sum_{|\bm{n}| \leq k} a_{\bm{n}} e^{i\epsilon n_1^2} \ket{\bm{n}}}{\sqrt{\sum_{|\bm{j}| \leq k} |a_{\bm{j}}|^2}},
    \end{eqnarray}
    where $|\bm{n}| = \sum_{i=1}^m n_i$. Then from Lemma \ref{lem:efficient_description}, we have that
    \begin{eqnarray}
        d(\psi,\psi_{\mathrm{red}}) = \frac12||\ket{\psi} \bra{\psi} - \ket{\psi_{\mathrm{red}}} \bra{\psi_{\mathrm{red}}}||_1 &\leq& \sqrt{\frac{s^E}{s^k}} ,\nonumber \\ d(\phi,\phi_{\mathrm{red}}) = \frac12||\ket{\phi} \bra{\phi} - \ket{\phi_{\mathrm{red}}} \bra{\phi_{\mathrm{red}}}||_1 &\leq& \sqrt{\frac{s^E}{s^k}}.
    \end{eqnarray}
    To find the trace distance between $\ket{\psi_{\mathrm{red}}}$ and $\ket{\phi_{\mathrm{red}}}$, we note that
    \begin{equation}
        \braket{\psi|\phi} = \frac{\sum_{|\bm{n}|\leq k} |a_{\bm{n}}|^2 e^{i\epsilon n_1^2} }{\sum_{|\bm{j}| \leq k} |a_{\bm{j}}|^2}.
    \end{equation}
    Therefore, the fidelity is given by
    \begin{equation}
        |\braket{\psi|\phi}|^2 = \frac{\sum_{|\bm{m}|,|\bm{n}| \leq k} |a_{\bm{m}}|^2 |a_{\bm{n}}|^2 e^{i \epsilon (n_1^2 - m_1^2)}}{\sum_{|\bm{p}|,|\bm{q}| \leq k} |a_{\bm{p}}|^2 |a_{\bm{q}}|^2}.
    \end{equation}
    Note that
    \begin{eqnarray}
        \sum_{|\bm{m}|,|\bm{n}| \leq k} |a_{\bm{m}}|^2 |a_{\bm{n}}|^2 e^{i \epsilon (n_1^2 - m_1^2)} &=& \sum_{\bm{m}= \bm{n}} |a_{\bm{m}}|^4 + \sum_{\bm{m} \neq \bm{n}} |a_{\bm{m}}|^2 |a_{\bm{n}}|^2\cos(\epsilon (n_1^2 - m_1^2)) \nonumber \\ &=& \sum_{|\bm{m}|,|\bm{n}|\leq k} |a_{\bm{m}}|^2 |a_{\bm{n}}|^2 + \sum_{\bm{m}\neq \bm{n}} |a_{\bm{m}}|^2 |a_{\bm{n}}|^2 (\cos(\epsilon (n_1^2 - m_1^2)) - 1). \nonumber \\
    \end{eqnarray}
    Therefore, we get
    \begin{equation}
        |\braket{\psi_{\mathrm{red}}|\phi_{\mathrm{red}}}|^2 = 1 + \frac{\sum_{\bm{m}\neq \bm{n}} |a_{\bm{m}}|^2 |a_{\bm{n}}|^2 (\cos(\epsilon (n_1^2 - m_1^2)) - 1)}{\sum_{|\bm{m}|,|\bm{n}|\leq k} |a_{\bm{m}}|^2 |a_{\bm{n}}|^2}
    \end{equation}
    For pure states, the trace distance between $\psi_\mathrm{red}$ and $\phi_\mathrm{red}$ is given by
    \begin{eqnarray}
    d(\psi_{red},\phi_{\mathrm{red}})^2 = 1 -  |\braket{\psi_{\mathrm{red}}|\phi_{\mathrm{red}}}|^2 =  \frac{\sum_{\bm{m}\neq \bm{n}} |a_{\bm{m}}|^2 |a_{\bm{n}}|^2 (1 - \cos(\epsilon (n_1^2 - m_1^2)))}{\sum_{|\bm{m}|,|\bm{n}|\leq k} |a_{\bm{m}}|^2 |a_{\bm{n}}|^2},
    \end{eqnarray}
    Using
    \begin{equation}
    1 -\cos x \leq \frac{x^2}{2},
    \end{equation}
    we write
    \begin{equation}
        d(\psi_{red},\phi_{\mathrm{red}})^2 \leq \frac{\epsilon^2 \sum_{\bm{m}\neq \bm{n}} |a_{\bm{m}}|^2 |a_{\bm{n}}|^2 (n_1^2 - m_1^2)^2}{2\sum_{|\bm{m}|,|\bm{n}|\leq k} |a_{\bm{m}}|^2 |a_{\bm{n}}|^2} =  \frac{\epsilon^2 \sum_{|\bm{m}|,|\bm{n}|\leq k} |a_{\bm{m}}|^2 |a_{\bm{n}}|^2 (n_1^2 - m_1^2)^2}{2\sum_{|\bm{m}|,|\bm{n}|\leq k} |a_{\bm{m}}|^2 |a_{\bm{n}}|^2}.
    \end{equation}
    Since $(n_1^2 - m_1^2)^2 \leq k^4$,
    \begin{equation}
       d(\psi_{red},\phi_{\mathrm{red}})^2 \leq \frac{\epsilon^2 k^4}{2} ,
    \end{equation}
    and hence
    \begin{equation}
    d(\psi_{red},\phi_{\mathrm{red}}) \leq \frac{\epsilon k^2}{\sqrt2}.
    \end{equation}
    Combined with the triangle inequality, we get
    \begin{equation}
        d(\psi,\phi) \leq 2 \sqrt{\frac{s^E}{s^k}} + \frac{\epsilon k^2}{\sqrt 2}.
    \end{equation}
Choosing the Fock state cutoff $k = \log_s(s^E/\epsilon^2)$, we get 
\begin{eqnarray}
   d(\psi,\phi) &\leq& 2 \epsilon + \frac{\epsilon}{\sqrt2} \log_s^2\left(\frac{s^E}{\epsilon^2}\right) \nonumber \\
   &=& \epsilon \left(2 + \frac{E^2}{\sqrt{2}} + \frac{\log_s^2(1/\epsilon^2)}{\sqrt{2}} + \sqrt{2}E\log_s(1/\epsilon^2)\right) \nonumber \\
   &=& \mathcal O\left(\epsilon E^2 \log_s^2 (1/\epsilon^2)\right).
\end{eqnarray}
\end{proof}
\section{Classical simulation of bosonic circuits with self-Kerr gates via phase-shifter decomposition}\label{appsec:kerr_gate_simul}
We first give the formal statement of complexity of classical simulation of bosonic circuits with self-Kerr gates, underlined by the following Theorem \ref{theo:kerr_gate_simul}:
\begin{theo}\label{apptheo:kerr_gate_simul}
     Given the $m$-mode vacuum state $\ket{0}^{\otimes m}$ evolving through a quantum circuit characterized by the unitary 
    \begin{equation}
        \hat U = \hat G_L \hat \kappa(x_L) \dots \hat G_1 \hat \kappa(x_1),
    \end{equation}
    where $\hat G_1,\dots,\hat G_L$ are arbitrary Gaussian gates, $\hat \kappa (x_1),\dots,\hat \kappa (x_L)$ are self-Kerr gates, then $\forall k \in (1,\dots,m)$ probabilities generated by heterodyne measurements on the first $k$ modes of the output state
    \begin{equation}
        P (\bm\alpha_k) = \left|\bra{\alpha_1 \alpha_2 \dots \alpha_k}\otimes \mathbb{I}_{m-k}\hat U\ket{0}^{\otimes m}\right|^2
    \end{equation}
    \begin{itemize}
        \item Can be computed exactly in polynomial time for log depth circuits $L = \mathcal O(\log(\poly(m)))$ if all the self-Kerr gates have rational parameters i.e.
        \begin{equation}
            x_i = p_i/q_i; \hspace{5mm} p_i \in \Z,q_i \in \Z_+^*, \forall i \in (1,\dots,L).
        \end{equation}
        \item Can be computed up to additive precision $\delta = 1/\mathcal O(\poly(m))$ in quasipolynomial time for log depth circuits $L = \mathcal O(\log(\poly(m)))$ for generic self-Kerr gates.
    \end{itemize}
    This allows for approximate strong simulation and in particular, approximate sampling
\end{theo}

\noindent We split the results into two parts: in Section \ref{appsubsec:rational}, we give the formal statement as well as the proof about the simulation of bosonic circuits with rational self-Kerr gates, whereas in Section \ref{appsubsec:irrational}, we give the formal statement as well as the proof about the simulation of bosonic circuits with possibly irrational self-Kerr gates.
\subsection{Classical simulation of bosonic circuits with rational self-Kerr gates}\label{appsubsec:rational}
We have the following Theorem regarding classical simulation of bosonic circuits with rational self-Kerr gates:
\begin{theo}\label{apptheo:rational_Kerr_simulation}
    Given a quantum state generated by the action of $L$ layers of Gaussian and self-Kerr gates
    \begin{equation}
       \hat U_L = \hat G_L \hat \kappa( p_L/q_L)\dots  \hat G_1 \hat \kappa(p_1/q_1),
    \end{equation}
    where $p_i\in \Z,q_i \in \Z_+^*$, $\forall i \in (1,\dots,L)$ (without the loss of generality, we can assume $q_i \geq 0, \forall i \in (1,\dots,L)$), then the probability distribution by heterodyne detection on the first $k$ modes of the output state
    \begin{equation}
        P(\bm \alpha_k) = \frac{1}{\pi^k}\left|\bra{\alpha_1\alpha_2 \dots\alpha_k}\hat U_L \ket{0}^{\otimes m}\right|^2
    \end{equation}
    can be computed in time $\mathcal O(m^4 q^{2L}) \forall k \in (1,\dots,m)$, where 
    \begin{equation}
        q = \max_{i \in (1,\dots,L)} q_i.
    \end{equation}
\end{theo}
\begin{proof}
We first describe the time complexity of describing the action of a single layer of Gaussian and self-Kerr gate
   \begin{equation}
        \hat G_i \hat \kappa (p_i/q_i)
   \end{equation}
   on a pure Gaussian state $\ket{G}$. By Lemma \ref{lem:kerr_decomp} and the Bloch--Messiah decomposition, we write
   \begin{eqnarray}
        \hat G_i \hat \kappa ( p_i/q_i) \ket{G} &=& \sum_{j=0}^{q_i - 1} g_{j,p_i,q_i} \hat G_i \hat R_1(2\pi j /q_i)\ket{G} \nonumber \\
        &=& \sum_{j=0}^{q_i - 1} g_{j,p_i,q_i} \hat V_i \hat D(\bm{\alpha}_i)\hat S(\bm{r}_i)\hat U_i \hat R_1(2\pi j /q_i) \ket{G},
   \end{eqnarray}
   where $\hat R_1 (\theta) \coloneqq e^{i\theta\hat N_1}$. Furthermore, from \cite{Reck1994}, any $m$-mode energy-preserving Gaussian unitary can be decomposed into $\mathcal{O}(m^2)$ phase-shifters and beam splitters. Now, from \cite[Theorem III.7]{Dias2024}, the action of a phase-shifter/beamsplitter/single-mode displacement operator on a Gaussian state can be described in time $\mathcal O(m)$, whereas the action of a single-mode squeezing can be described in time $\mathcal O(m^3)$. This gives the time complexity to describe the action of $\hat V_i \hat D(\bm{\alpha}_i)\hat S(\bm{r}_i)\hat U_i \hat R_1(2\pi j /q_i)$ on $\ket{G}$ to be $\mathcal O(m^3 + m^2 + m^4 + m^3 + m) =\mathcal O(m^4), \forall j$. We have to describe $q_i$ such actions, this gives the time complexity of describing $ \hat \kappa ( p_i/q_i) \hat G_i \ket{G}$ as a superposition of $q_i \leq q$ Gaussian state to be $\mathcal O (q_i m^4) = \mathcal O(qm^4)$. 

   Starting with the $m$-mode Gaussian state vacuum, we see from the previous discussion that a single layer of Gaussian and self-Kerr gate convert a Gaussian state into a superposition of $q$ Gaussian states. Therefore, after $i$ layers of Gaussian and self-Kerr gates
   \begin{equation}
       \ket{\psi_i} = \hat G_i  \hat \kappa ( p_i/q_i) \dots  \hat G_1 \hat \kappa ( p_1/q_1)
   \end{equation}
   can be written as a superposition of $q^i$ Gaussian states. We now give the time complexity of describing the action of $i+1$'th layer of Gaussian and self-Kerr gate on $\ket{\psi_i}$ given the description of $\ket{\psi_i}$ as a superposition of $q^i$ Gaussian states
   \begin{equation}
       \ket{\psi_i} = \sum_{k=1}^{q^i} c_k \ket{G_k},
   \end{equation}
   such that
 \begin{eqnarray}
        \hat \kappa ( p_{i+1}/q_{i+1}) \hat G_{i+1} \ket{\psi_i} &=& \sum_{k=1}^{q^i} c_k \hat \kappa (p_{i+1}/q_{i+1}) \hat G_{i+1} \ket{G_k}.
   \end{eqnarray}
   From the preceding discussion, the description of $ \hat G_{i+1} \hat \kappa ( p_{i+1}/q_{i+1}) \ket{G_k}$ as a superposition of $\mathcal O(q)$ Gaussian states can be obtained in time $\mathcal O(q m^4)$. There are $q^{i}$ such descriptions. This gives the time to describe $\hat \kappa ( p_{i+1}/q_{i+1}) \hat G_{i+1} \ket{\psi_i}$ as a superposition of $q^{i+1}$ Gaussian states to be $\mathcal O(q^{i+1}m^4)$. We do this process for $i \in (0,\dots,L-1)$ and this gives the time complexity of describing $\ket{\psi} = \hat \kappa (p_L/q_L) \hat G_L \dots \hat \kappa ( p_1/q_1) \hat G_1 \ket{0}^{\otimes m}$ as a superposition of $\mathcal O(q^L)$ Gaussian states
   \begin{equation}
       \ket{\psi} = \sum_{k=1}^{q^L} C_k \ket{G_k}
   \end{equation}
   in time $\mathcal{O}(Lm^4 (q + q^2 + q^3+\dots+q^L)) = \mathcal O(q^L m^4L)$. And once we have the description of $\ket{\psi}$, from \cite[Theorem III.8]{Dias2024} and \cite[Lemma IV.1]{Dias2024} we can calculate
   \begin{equation}
       P(\bm{\alpha}_k) = \frac{1}{\pi^k}|\braket{\alpha_1\dots\alpha_k|\psi}|^2
   \end{equation}
   in time $\mathcal O(q^{2L}m^3)$. This gives the total time complexity of the exact computation of probabilities generated by heterodyne detection on output state of the circuit composed of $m$-mode vacuum state fed to $L$ layers of Gaussian and rational self-Kerr gates to be $\mathcal{O}(q^{2L}m^3 + q^L m^4) = \mathcal{O}(q^{2L}m^4)$
\end{proof}
\noindent Therefore we see that Theorem \ref{apptheo:rational_Kerr_simulation} allows for efficient classical simulation of log depth $L = \mathcal{O}(\poly(m))$ bosonic circuits with rational self-Kerr gates.
\subsection{Classical simulation of Bosonic circuits with irrational self-Kerr gates}\label{appsubsec:irrational}
We have the following Theorem regarding classical simulation of bosonic circuits with possibly irrational self-Kerr gates:
\begin{theo}\label{apptheo:irrational_Kerr_simul}
        Given a quantum state generated by the action of $L = \mathcal O(\log(\poly(m)))$ layers of Gaussian and self-Kerr gates
    \begin{equation}\label{appeq:irrational_Kerr_circuit}
       \hat U_L = \hat G_L \hat \kappa(x_L) \dots \hat G_1 \hat \kappa(x_1),
    \end{equation}
    on the $m$-mode vacuum state, such that $c$ out of the $L$ self-Kerr gates have irrational parameters (up to $\pi$),
    then the probability distribution generated by heterodyne detection on the first $k$ modes of the output state
    \begin{equation}
        P(\bm \alpha_k) = \frac{1}{\pi^k}\left|\bra{\alpha_1\alpha_2 \dots\alpha_k}\hat U_L \ket{0}^{\otimes m}\right|^2
    \end{equation}
    can be computed up to $\delta = 1/\mathcal O(\poly(m))$ additive precision in quasipolynomial time.
\end{theo}
\begin{proof}
Denoting
\begin{equation}
    \ket{\psi_{i}} = \hat G_i \hat \kappa (x_i) \dots \hat G_1 \hat \kappa (x_1) \ket{0}^{\otimes m}, \hspace{5mm} \forall i \in (1,\dots,L),
\end{equation}
we first note from Theorem \ref{theo:exp_energy_bound_simplified} that we can find a $t_i > 1$ and $E_i \in \R_{+}$ with scalings given by
\begin{equation}
    1/\log(t_i) = \mathcal O (\exp(L)), \hspace{5mm} E_i = \mathcal O(m\exp(L))
 \end{equation}
 and such that 
 \begin{equation}
     \bra{\psi_i} t_i^{\hat N} \ket{\psi_i} < t_i^{E_i}.
 \end{equation}
  With this knowledge, the proof of the Theorem proceeds as follows: Given the description of the circuit (Eq.~\ref{appeq:irrational_Kerr_circuit}), we first recognize the set of self-Kerr gates with irrational parameters up to $\pi$ and for each of these irrational self-Kerr gate parameters $x_i$, we use the Diophantine approximation with continued fractions \cite{hurwitz1891ueber} to approximate $x_i$ such that
   \begin{equation}
       \left|\pi x_i - \pi \frac{p_i}{q_i} \right| < \frac{\pi}{\sqrt 5q_i^2}=: \epsilon_i.
   \end{equation}
   To characterize the error accumulated by approximating the irrational self-Kerr gates as rational self-Kerr gates throughout the circuit, consider the following: Given two states $\ket{\psi_{i-1}}$ and $\ket{\phi_{i-1}}$ as input to the self-Kerr gate $\hat \kappa (x_i)$ ,we want to upper bound the trace distance between $\hat \kappa(x_i) \ket{\psi_{i-1}}$ and $\hat{\kappa}(\pi p_i/q_i) \ket{\phi_{i-1}}$. We write
   \begin{eqnarray}
       &&\frac12\|\hat \kappa(x_i) \ket{\psi_{i-1}} \bra{\psi_{i-1}}\hat\kappa^\dagger(x_i)  - \hat{\kappa}(p_i/q_i) \ket{\phi_{i-1}}\bra{\phi_{i-1}} \hat \kappa^\dagger ( p_i/q_i)\|_1  \nonumber \\
       &&\hspace{5mm}=\frac12\|\hat \kappa (\epsilon_i)\hat \kappa( p_i/q_i) \ket{\psi_{i-1}} \bra{\psi_{i-1}}\hat\kappa^\dagger(p_i/q_i) \hat \kappa^\dagger(\epsilon_i)  - \hat{\kappa}( p_i/q_i) \ket{\phi_{i-1}}\bra{\phi_{i-1}} \hat \kappa^\dagger ( p_i/q_i)\|_1 \nonumber \\
       &&\hspace{5mm}\leq \frac12\|\hat \kappa (\epsilon_i)\hat \kappa( p_i/q_i) \ket{\psi_{i-1}} \bra{\psi_{i-1}}\hat\kappa^\dagger( p_i/q_i) \hat \kappa^\dagger (\epsilon_i)  - \hat{\kappa}( p_i/q_i) \ket{\psi_{i-1}}\bra{\psi_{i-1}} \hat \kappa^\dagger ( p_i/q_i)\|_1 \nonumber\\
       &&\hspace{8mm}+\frac12\|\hat \kappa( p_i/q_i) \ket{\psi_{i-1}} \bra{\psi_{i-1}}\hat\kappa^\dagger( p_i/q_i)   - \hat{\kappa}(p_i/q_i) \ket{\phi_{i-1}}\bra{\phi_{i-1}} \hat \kappa^\dagger ( p_i/q_i)\|_1 \nonumber \\
       &&\hspace{5mm} = \epsilon_{\mathrm{prior}} + \frac12 \|\ket{\psi_{i-1}}\bra{\psi_{i-1}} - \ket{\phi_{i-1}}\bra{\phi_{i-1}}\|_1,
   \end{eqnarray}
   where $\epsilon_{\mathrm{prior}} = \mathcal O(\epsilon_{i-1} E_{i-1}^2 \log_{t_{i-1}}^2(1/\epsilon_{i-1}^2 ))$. For the first inequality, we use the triangle inequality whereas for the second inequality, we use Lemma \ref{lem:eps_kerr_gate_td} for the first term and the fact that unitary gates do not change the trace distance for the second inequality. Choosing $\epsilon_i = \epsilon, \forall i \in (1,\dots,L)$, then from Theorem \ref{theo:exp_energy_bound_simplified} and its subsequent proof in Section \ref{appsec:exp_energy_bound energy simplified}, its clear that $E_i \leq E_L$ and that $\log_{t_{i}}(1/\epsilon_i^2 )) \leq \log_{t_{L}}(1/\epsilon_i^2 ))$, $\forall i \in(1,\dots,L)$ (since $\log_x(a)$ is a decreasing function in x and the constant $t_i$ decreases with subsequent layers). Therefore, we see that whenever we are approximating an irrational self-Kerr gate by a rational one, we are adding a trace distance error of order $\mathcal O(\epsilon E_L^2 \log_{t_{L}}^2(1/\epsilon^2))$. We further upper bound $E_L^2 \log_{t_{L}}^2(1/\epsilon^2)$ by $E_L^3/\log^2(t_L)$ and later justify why this is a valid assumption. With $c$ such irrational self-Kerr gates, to ensure that the approximated output state $\ket{\Phi}$ is $\delta = 1/\mathcal O(\poly(m))$ close in trace distance to the actual output state $\ket{\psi} = \hat{U}_L \ket{0}^{\otimes m}$, we ensure that
   \begin{eqnarray}
       \epsilon E_L^3/\log^2(t_L)  &=& \frac{\delta}{c},\nonumber \\
       \frac{1}{\epsilon} = \frac{cE_L^3}{\delta\log^2(t_L)}.
   \end{eqnarray}
   Now, from Theorem \ref{theo:exp_energy_bound_simplified}, $E_L = \mathcal O(m \exp(L)) = \mathcal O(\poly(m))$ for $L = \log(\poly(m))$ and similarly, $1/\log(t_L) = \mathcal O(\exp(L)) = \mathcal O(\poly(m))$. With $c \leq L$ , we see that we require
   \begin{equation}
       \frac{1}{\epsilon} = \mathcal O(\poly(m)).
   \end{equation}
   This justifies our earlier assumption, since for $E_L,1/\epsilon = \mathcal O(\poly(m))$, $E_L^2 \log_s^2(1/\epsilon^2) \leq E_L^3/\log^2(s)$ since polynomial grows faster than a logarithmic. With this choice of $\epsilon$, the Diophantine approximation requires the denominator of the approximated rational self-Kerr gates to be such that
   \begin{equation}
       q^2  = \frac{\pi}{\sqrt{5}\epsilon} 
   \end{equation}
   for the output state $\ket{\Phi}$ to be $\delta = 1/\mathcal O(\poly(m))$ close in trace distance to the actual output state. Furthermore, from the property of trace distance, defining
   \begin{equation}
       \tilde P(\bm\alpha_k) = \frac{1}{\pi^k} \left|\braket{\alpha_1\alpha_2\dots\alpha_k|\Phi}\right|^2,
   \end{equation}
   it follows that
   \begin{equation}
       \left|P(\bm\alpha_k) - \tilde P(\bm\alpha_k)\right| \leq \frac12\|\ket{\psi}\bra{\psi} - \ket{\Phi}\bra{\Phi}\|_1 \leq \delta.
   \end{equation}
   And from Theorem \ref{apptheo:rational_Kerr_simulation}, $\tilde P(\bm\alpha_k)$ can be evaluated in time $\mathcal{O}(q^{2L}m^4)$, where $q = \max_{i \in (1,\dots,L)} q_i$. Assuming that the integers in the denominator of the originally rational self-Kerr gates are smaller than the integers in the denominator of the approximately rational self-Kerr gates, we get that $\tilde P(\bm\alpha_k)$ can be computed in time 
   \begin{equation}
       \mathcal O\left(m^4 \frac{\pi^L}{5^{L/2}\epsilon^L}\right) = \mathcal O(\poly(m)^L),
   \end{equation}
   Putting $L = \log(\poly(m))$, we see that we can compute (possibly marginal) probabilites upto inverse polynomial precision in quasipolynomial time.
   \end{proof}
\section{Bound of the 1-norms of the self-Kerr gate decomposition coefficients}\label{appsec:two-norm_bounds}
We have the following result on 1-norms of the self-Kerr gate decomposition coefficients:
\begin{theo}
Given the decomposition of a self-Kerr gate $\kappa(p/q)$ as a sum of $q$ phase-shifters (Lemma \ref{lem:kerr_decomp}), it follows that 
\begin{equation}
    \|g_q\|_1 = \sum_{j=0}^{q-1} |g_{j,p,q}^{(e/o)}| \leq \sqrt{q}.
\end{equation}
\end{theo}
\begin{proof}
     We show the proof for odd $q$, the proof for even $q$ is similar. Using the Cauchy-Scwartz inequality,
    \begin{equation}
        ||g_q^{(o)}||_1 = \sum_{j=0}^{q-1} |g_{j,p,q}^{(o)}| \leq \left(\sum_{j=0}^{q-1} |g_{j,p,q}^{(o)}|^2\right)^{1/2} \left(\sum_{j=0}^{q-1} 1\right)^{1/2} = \sqrt{q} \left(\sum_{j=0}^{q-1} |g_j^{(o)}|^2\right)^{1/2}.
    \end{equation}
    We note that
    \begin{eqnarray}
        \sum_{j=0}^{q-1} |g_{j,p,q}^{(o)}|^2 &=& \sum_{j=0}^{q-1}  g_{j,p,q}^{(o)*}g_{j,p,q}^{(o)} = \frac{1}{q^2} \sum_{j=0}^{q-1} \sum_{m,n=0}^{q-1} \exp\left(-i\pi \frac{p}{q}\left(n(n-1) - m(m-1)\right)\right)\exp\left(\frac{2\pi i j}{q} (n-m)\right) \nonumber \\
        &=& \frac{1}{q^2} \sum_{m,n=0}^{q-1} \exp\left(-i\pi \frac{p}{q}\left(n(n-1) - m(m-1)\right)\right) \sum_{j=0}^{q-1} \exp\left(\frac{2\pi i j}{q} (n-m)\right) \nonumber \\
        &=& \frac{1}{q} \sum_{m,n=0}^{q-1} \exp\left(-i\pi \frac{p}{q}\left(n(n-1) - m(m-1)\right)\right) \delta_{n-m,0} = \frac{1}{q} \times q = 1,
    \end{eqnarray}
    where in the third line, we are using Eq.~\ref{eq:kron_delta}. Therefore,
    \begin{equation}
        ||g_q^{(o)}||_1 \leq \sqrt{q}.
    \end{equation}
\end{proof}
\end{document}